\documentclass[12pt,a4paper]{article}

\usepackage{mathtools}
\usepackage{amsmath}
\usepackage{bbm}
\usepackage{siunitx}
\usepackage{multirow}
\usepackage{graphicx}
\usepackage{amssymb}
\usepackage[a4paper, margin=2.5cm]{geometry}
\usepackage{natbib}
\usepackage{amsthm}
\usepackage{hyperref}
\usepackage{fancyhdr}
\usepackage{appendix}
\usepackage{eso-pic}
\usepackage{tikz}

\graphicspath{{Fig/}}

\theoremstyle{plain}
\newtheorem{theorem}{Theorem}[section]
\newtheorem{proposition}{Proposition}[section]
\theoremstyle{definition}
\newtheorem{example}{Example}[section]
\newtheorem{definition}{Definition}[section]
\theoremstyle{remark}

\newcommand{\hide}[1]{}

\AddToShipoutPictureFG{%
  \begin{tikzpicture}[remember picture, overlay]
    \node[rotate=45, scale=8, text=gray, opacity=0.5] at (current page.center) {\textbf{PREPRINT}};
  \end{tikzpicture}
}

\begin{document}

\title{A duration-augmented binary Markov chain for rainfall occurrence with long dry spells}

\author{%
  \normalsize
  Antoine Doizé \textsuperscript{1,*},
  Denis Allard\textsuperscript{2},
  Philippe Naveau\textsuperscript{3},
  Olivier Wintenberger\textsuperscript{4}
  \\[0.4em]
  \textsuperscript{1}LPSM, Sorbonne Université, Paris, France \\
  \textsuperscript{2}BioSP, INRAE, 84914 Avignon, France \\
  \textsuperscript{3}LSCE, CNRS-CEA-UVSQ, Gif sur Yvette, France \\
  \textsuperscript{4}LPSM, Sorbonne Université, Paris, France
 \\[0.2em]
  \textsuperscript{*}{\footnotesize\href{mailto:antoine.doize@sorbonne-universite.fr}{antoine.doize@sorbonne-universite.fr}}
} 

\date{\small Preprint --- May 7, 2026}

\maketitle

\begin{abstract}
Simulating realistic wet and dry spells is central in weather generators and climate-impact studies. While finite-order Markov chains are standard, they often fail to reproduce persistent dry conditions due to their inherent subexponential decay. We model rainfall occurrence by introducing a duration-augmented binary Markov chain. We establish a link with alternating renewal chains, enabling flexible parametric modelling of wet and dry spell duration distribution. We model those using two regime-adapted specifications from the general class of extended Generalized Pareto Distributions, yielding flexible tail behaviour across various climates. We use estimation methods adapted to each specification. Our model is applied to around 200 stations in the South of Europe spanning diverse Mediterranean and continental climates. We compare this framework to standard Markov models in characterising persistence and high-quantile extrapolation. The approach is generic, extending naturally to multi-state settings or other binary sequence applications in environmental statistics.

\medskip\noindent\textbf{Keywords:} Extreme dry spells, Rainfall occurrence, Stochastic generator
\end{abstract}

\section{Introduction}\label{sec:introduction}

Rainfall data underpin analyses across many fields, including ecosystem dynamics \citep{Friend1997_Hybrid_v3_EcolModel}, pollution dispersion for prevention \citep{Knisel1980_CREAMS_USDA}, urban hydrology, flash-flood insurance \citep{2019_leal}, and crop-yield sensitivity to rainfall variability \citep{2017_muneepeerakul_rainfall_intensity_frequency}. Among these applications, extreme rainfall events constitute primary risk drivers, in particular in the context of prolonged dry spells. Indeed, such events suppress crop yields \citep{2005_richter_wheat_drought}, diminish ecosystem resilience \citep{2003_fay_altered_rainfall_grassland}, and reduce river flow levels, with consequent adverse effects on wetland habitat integrity \citep{rolls2012mechanistic}, deterioration of water quality, and impairment of power plant cooling system efficiency. Despite this broad relevance, modelling the statistical distribution of rainfall occurrences and specifically the tail of the distribution, remains a fundamental challenge in hydrology and climate science, due to internal variability and complex multi-scale physical processes. This motivates the use of generative algorithms. In the statistical climatology literature \cite[see, e.g.][]{2015_ailliot_overview_weather_type_models}, these sampling techniques have been called 
stochastic weather generators. These methods are used to to generate plausible sequences and to enable extrapolation of extreme events that have not yet occurred but remain statistically plausible. The latter is usually performed by producing long simulations.

Numerous methodologies have been proposed for stochastic weather generators \citep{1999_wilks_weather_generation_game_review, 2017_olson_kleiber}. Yet a statistical challenge in modelling daily rainfall arises from the fundamental dichotomy between dry and wet days, leading to a discrete probability at zero (dry days) and continuous positive values (wet days). In this study, we focus on daily rainfall occurrences at a single site modelled by chain-dependent processes. For the sake of completeness, we briefly comment now on other methodologies. Alternative strategies include resampling approaches \citep{2003_harrold_nonparametric_model_daily_rainfall, wilby2003multi}, which are limited by their inability to generate values beyond recorded extremes, and censored continuous models, such as censored Gaussian formulations \citep{1995_hutchinson_stochastic_space_time_weather, allard2015disaggregating, 2015_baxevani_spatiotemporal_precipitation_generator}. The latter offer flexibility, allowing spatial extensions, Bayesian integration \citep{2018_benoit_stochastic_rainfall_subkilometer}, and covariate-dependent thresholds \citep{2002_qian_gaussian_censored_covariate_threshold}. However, as highlighted in \citet{2015_baxevani_spatiotemporal_precipitation_generator}, these models are not tuned to reproduce the longest dry spells (see right panels of Fig.~10 and Fig.~19 of the reference).

Chain-dependent processes are typically structured as a two-step methodology: the first step models rainfall occurrences, while the second models rainfall intensities conditional on rain occurrences. Early developments in this line of research employed two-state first-order Markov chains to represent rainfall occurrences \citep{1962_gabriel_markov_chain_daily_rainfall_telaviv}. Subsequent work extended this framework in several directions, including the incorporation of seasonality via Fourier analysis \citep{1981_richardson_stochastic_simulation_weather} or spatial extensions \citep{1991_zucchini_hidden_markov_space_time_precipitation}.
Despite these advances, a number of limitations have been consistently highlighted. First, Markov models tend to underestimate the variance of spell durations \citep{1999_wilks_interannual_variability_extremes_sdp}. More importantly, they systematically underestimate the frequency of extreme dry spells, a problem identified as early as \citet{1964_hopkins_statistics_daily_rainfall_occurrence} and repeatedly emphasised in subsequent studies \citep{1991_racsko_serial_local_stochastic_weather_models}. This shortcoming stems from the fact that Markov chains imply an exponentially decaying distribution for spell durations, which contrasts with empirical evidence from observed records when it comes to dry spells. To illustrate this limitation, we examine daily precipitation records from the European Climate Assessment \& Dataset (ECA\&D), whose data and preprocessing are detailed in Section~\ref{sec:application_europe_data}. The stations considered are displayed on the map in Fig.~\ref{fig:map_stations}. For readability, the station-focused figures throughout the article are shown only for Palermo (orange dot). The figures for the rest of the stations (black dots) are available in the provided code repository. The sequence of dry spell durations (properly defined in Section~\ref{sec:application_europe_data}) is extracted. Then, the survival function of dry spell durations in Palermo, which provides the probability of a dry spell duration to be greater than a given duration, is estimated and displayed in Fig.~\ref{fig:survival_func_sample_data_intro}, for every season. The curves which are displayed are the empirical survival function (orange), the geometric fit (red curve) which follows from a two-state first-order Markov chain model, and our suggested fitted model (black curve) which will be detailed in Section~\ref{subsec:spell_duration_distribution_specification}. The geometric survival function has an exponential tail, which may be adapted to some climates, for instance in Palermo-summer (top right panel). However, there are some stations and seasons for which this model fails to catch heavy-tailed dry spell distributions, which is striking here for Palermo-spring (top left panel) where the tail of the geometric survival declines at a markedly faster rate than its empirical counterpart. Thus this modelling significantly underestimates the probability of long dry spells, and the gap that widens further into the tail. Consequently, extrapolations from such a model will severely underestimate extreme dry events. The suggested survival function is much closer to the empirical fit in the tail of the distribution, which demonstrates its ability to model heavy-tailed survival functions for spell durations. Such pattern is common in spring (and can apply to autumn and winter for other stations): a transitional season, it exhibits a wide range of durations from short spells to extremely long ones.
\begin{figure}[!htbp]
    \centering
    \includegraphics[trim=0 0 0 0,clip,width=0.9\textwidth]{ 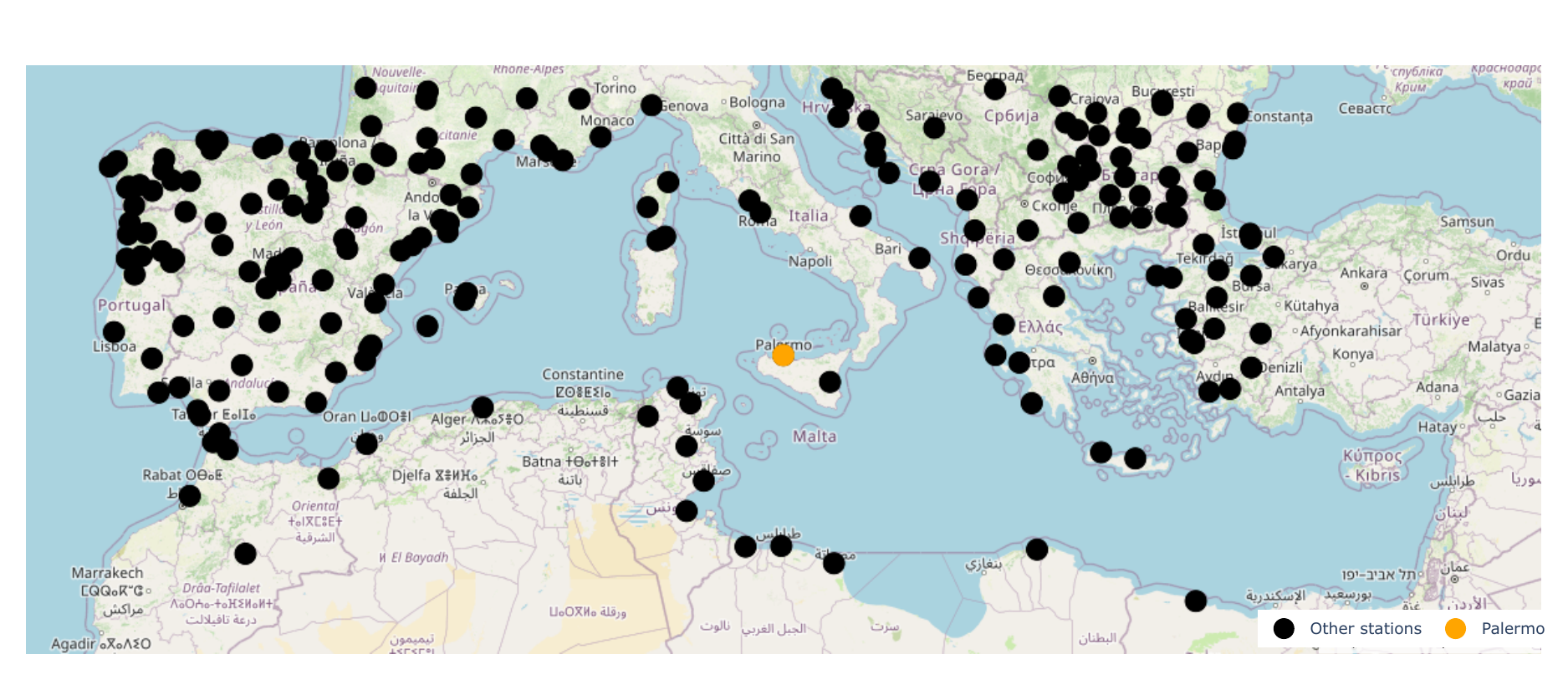}
    \caption{Map of the stations extracted from the ECA\&D. Rainfall data were extracted on the period 1945 - 2025 (data processing detailed in Section~\ref{sec:application_europe_data}). Throughout the article, station-specific figures are shown only for Palermo (orange dot) and analogous figures for the remaining stations (black dots) are available in the provided code repository.} \label{fig:map_stations}
\end{figure}
\begin{figure}[!htbp]
    \centering
    \includegraphics[width=0.75\linewidth]{ 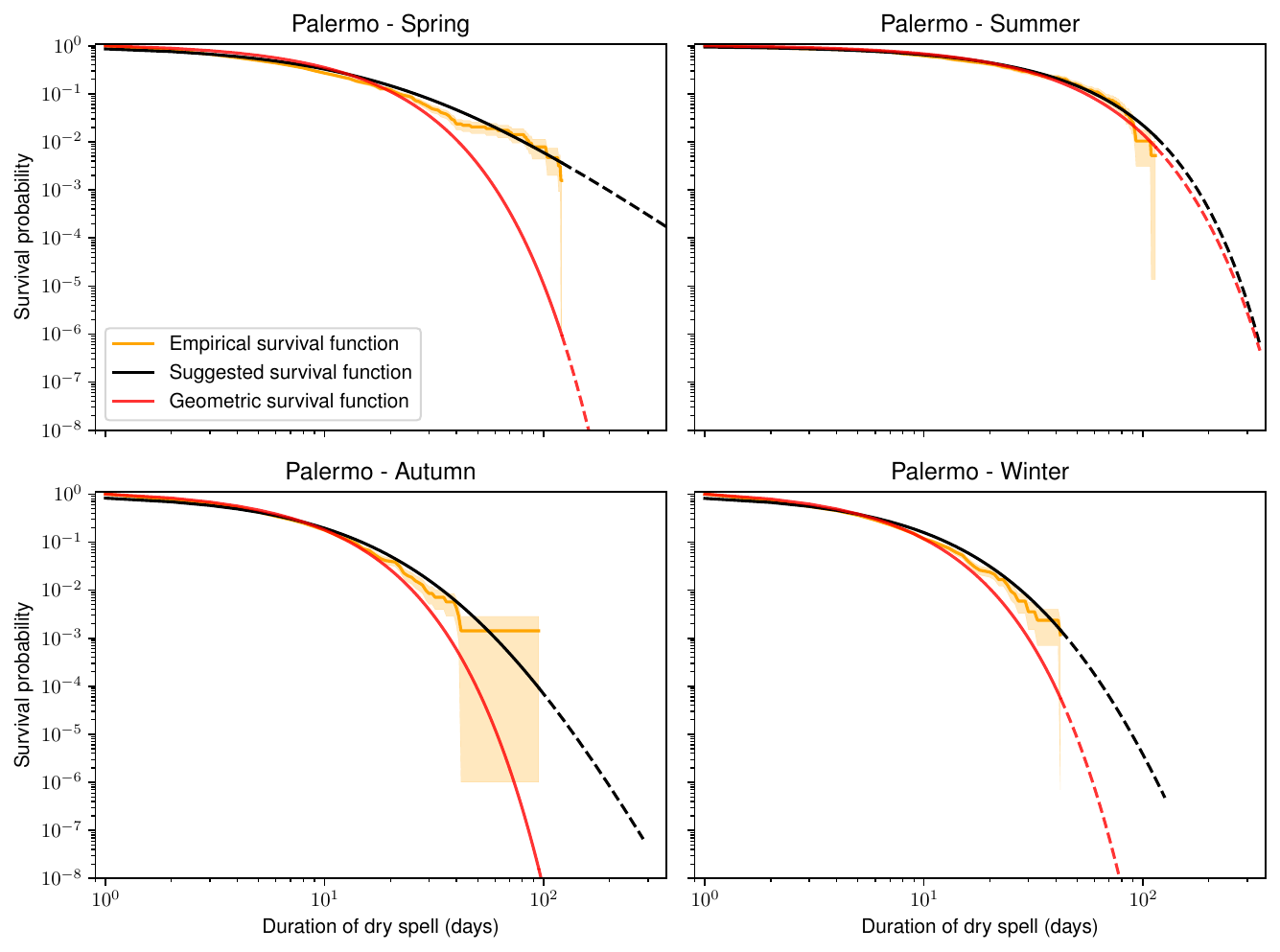}
    \caption{Survival function of dry spell durations in Palermo (log-log scale). Empirical estimation (orange curve) is compared to a geometric fit (red curve), and the suggested fit (black curve) detailed in Section~\ref{subsec:spell_duration_distribution_specification} and supported throughout this article.} \label{fig:survival_func_sample_data_intro}
\end{figure}
In response, numerous modelling strategies have been suggested to capture a broader class of distributions, including the use of higher-order Markov models \citep{1984_stern_model_fitting_daily_rainfall} and hidden Markov models with more than two hidden states. These developments laid the groundwork for modern weather generator frameworks \citep{2010_flecher_stochastic_weather_generator_skewed}, reviewed in \citet{2015_ailliot_overview_weather_type_models}. The underlying rationale is to represent distinct weather types, enabling for instance the differentiation between regular dry spells and severe dry spells \citep{2020_stoner_advanced_hmm_hourly_rainfall}.
Nonetheless, two key limitations can be identified. First, the introduction of complex models with hyperparameters, such as the number of hidden states, poses significant model selection challenges. The choice of model may vary depending on the selection criterion employed and may also differ across sites with distinct climatic regimes \citep{2008_schoof_order_markov_chain_daily_precipitation}, thereby complicating spatial generalisation. Second, difficulties persist in adequately representing extreme dry spell duration distributions. Even with higher-order Markov chains, the underlying model structure enforces geometrically distributed spell durations \citep{2008_lennartsson_precipitation_sweden_markov_composite}, which appears to  be too restrictive to describe the extreme behaviour of dry spell duration at different sites. These limitations have motivated the exploration of different approaches for modelling rainfall occurrence.

Within rainfall occurrence modelling, one option is to focus on modelling spell durations directly. This idea leverages alternating renewal processes from applied probability \citep{1987_asmussen_applied_probability_queues} and is often referred to as spell-length models in hydrometeorology \citep{1999_wilks_weather_generation_game_review}. These methods require independence between successive spell durations, a weaker assumption that for Markov models, particularly relevant for long dry spells. Their key advantage lies in the explicit representation of spell duration distributions, which avoids exponentially decaying tails. These approaches have been investigated since early studies, initially through attempts to discretize continuous alternating renewal processes into two-state \citep{1964_green_model_rainfall_occurrence} or three-state models \citep{1990_hutchinson_three_state_markov_point_rainfall}. However, such continuous-time representations quickly become mathematically complex, even when restricting to exponential durations.\hide{ Extensions include Markov renewal models \citep{1987_foufoula_markov_renewal_rainfall_occurrences} and Markov-Bernoulli processes \citep{1987_smith_markov_bernoulli_process}, thereby establishing links with point process theory.} In parallel, a different strand of research focuses on fitting discrete distributions, in particular the geometric distribution, to recorded spell durations. However, early studies demonstrated persistence in dry spells, with the probability of remaining in a dry spell increasing with spell duration \citep{1966_chatfield_wet_dry_spells}, a behaviour inconsistent with the memoryless property of the geometric distribution. To address this, more flexible distributions were employed, such as the log-series \citep{1970_green_generalized_probability_wet_dry}, gamma \citep{1973_quelennec_contribution_etude_probabiliste_pluies}, negative binomial \citep{1977_buishand_stochastic_modelling_daily_rainfall,1981_galloy_model_weather_cycles_rainfall_occurrence}, mixtures of geometric and log-series distributions \citep{2010_deni_probability_models_dry_wet_spells}. Subsequent extensions incorporated Fourier analysis to capture smooth seasonality \citep{1982_roldan_woolhiser_stochastic_daily_precipitation_occurrence,1991_racsko_serial_local_stochastic_weather_models}.

To our knowledge, the current state of research exhibits two blind spots. The first concerns the treatment of extreme spell durations. Although frequently identified as a critical challenge, neither Markov models nor alternating renewal models dealt specifically with this issue, and yet they are commonly used in climate impact studies \citep{2014_chen_review_stochastic_generation_precipitation} along with resampling methods \citep{2008_semenov_simulation_extreme_weather_swg}. Some studies have addressed extreme dry spells by fitting Generalized Pareto (GP) or Generalized Extreme Value (GEV) distributions \citep{2006_lana_long_dry_spells_iberian}, but these approaches focus exclusively on extremes, without capturing the full distribution or wet spells, and are thus unsuitable for daily rainfall occurrence simulation.
The second blind spot concerns the relationship between alternating renewal chains and a wide class of Markov chains, which is detailed in this work. While this connection has been exploited in \cite[Section~3.8]{resnick1992adventures}, it has not been formally established.

This article addresses both gaps within the chain-dependent process framework, focusing on the daily rainfall occurrence modelling. We introduce a binary Markov chain whose state space is augmented by the current spell duration, that is the number of consecutive days already spent in the ongoing dry or wet spell. The assumption underlying this approach is mild:  successive spell durations are assumed independent, a weaker requirement than the finite-order Markov property, especially for long spells.
Two methodological contributions follow. First, the augmented state space gives explicit control over the spell duration distribution: any parametric family for dry and wet spell lengths can be embedded while preserving a Markov structure for simulation. The extended Generalized Pareto distribution (eGPD) class, introduced by \citet{2016_naveau_modelling_jointly_rainfall_intensities} and overviewed in \citet{2026_Naveau}, is used to model the bulk and the lower and upper tails of the spell distribution jointly. This addresses the first gap. The second one is dealt with by establishing a mathematical equivalence between the model and an alternating renewal chain, building on recent results by \citet{2025_kozubowski_waiting_time_discrete}. This duality yields a goodness-of-fit procedure tailored to the features that are central to our construction, and provides asymptotic results which are useful for risk assessment of extreme dry spells. The framework is applied to the ECA\&D station network across southern Europe and benchmarked against a two-state first-order Markov chain baseline. A map highlights stations where our model predicts higher exposure to long dry spells than the baseline.

This article is organised as follows. Section~\ref{sec:def_bmcd} introduces the Binary Markov Chain with Duration model, establishes its main theoretical properties, and describes the methodology for parameter estimation and model validation. Section~\ref{sec:application_europe_data} applies the model to a set of stations across southern Europe, assesses how well it fits the data, and demonstrates the improvement it affords for extreme dry spell risk assessment. Section~\ref{sec:conclusion} summarises the main improvements, discusses the assumptions underlying the approach, and outlines directions for future work.

\section{Binary Markov Chain with Duration}\label{sec:def_bmcd}

Denote $\mathbb{N}^*=\mathbb{N}\setminus\{0\}$.
For $a\in\mathbb{R}$, we denote $a_+ := \max(a,0)$. We further denote
$\mathbf{a}=(a_i)_{i\in\mathbb{N}^*}$ a sequence of probabilities with $a_i\in[0,1]$ for all $i\in\mathbb{N}^*$, and for any integers $p_1,p_2 \ge 1$, $\;\mathbf{a}_{|p_1:p_2}=(a_i)_{p_1 \le i \le p_2}$ is the truncated vector of the elements $p_1$ to $p_2$. For a random variable $X$,  $\mathbb{E}[X]$ and $\mathrm{Var}(X)$ are its expectation and variance, respectively, and $F_X$, $\overline{F}_X=1-F_X$ are its cumulative distribution function (cdf) and survival function. When the distribution of $X$ is parametrised by a parameter $\theta$ belonging to some parameter space $\Theta$, we write $F_\theta$ and $\overline{F}_\theta$ in place of $F_X$ and $\overline{F}_X$, and denote by $\mathbb{P}_\theta$ and $\mathbb{E}_\theta$ the corresponding probability and expectation. When there is no ambiguity, we may suppress superscripts identifying the underlying random variable. Finally, we denote $\mathbbm{1}(\cdot)$ the indicator function, and $\lceil \cdot \rceil$ the ceiling function.

\begin{definition}[Binary Markov Chain with Duration]
\label{def:discrete_binary_markov_chain}
Let $\mathbf{q}^{(0)}$ and $\mathbf{q}^{(1)}$ be two sequences of probabilities. 
We call a random sequence $(R_n, D_n)_{n \geq 0}$ a \emph{Binary Markov Chain with Duration} (BMCD)  with \emph{exit probabilities} $(\mathbf{q}^{(0)}, \mathbf{q}^{(1)})$ if it is a Markov chain defined on 
 $\{0,1\} \times \mathbb{N}^*$ such that, 
 for every positive integer $n$, 
\[
(R_{n+1}, D_{n+1}) =
\begin{cases}
(1 - R_n, 1), & \text{with probability } q^{(R_n)}_{D_n}, \\[6pt]
(R_n, D_n + 1), & \text{with probability } 1 - q^{(R_n)}_{D_n}.
\end{cases}
\]
Without loss of generality, we set $(R_0, D_0) = (0,1)$.
\end{definition}
\noindent We also define the dry and wet spell durations as the hitting times
\begin{subequations}\label{eq:durations}
\begin{align}
\tau^{(r)}_k &\coloneqq \inf\Big\{\,n\ge \sum_{j=1}^{k-1}\big(\tau^{(0)}_j+\tau^{(1)}_j\big)
+ \mathbbm{1}_{\{r=1\}}\tau^{(0)}_k:\ (R_n,D_n)=(1-r,1)\,\Big\} \nonumber\\
&\quad -\Big(\sum_{j=1}^{k-1}\big(\tau^{(0)}_j+\tau^{(1)}_j\big)+ \mathbbm{1}_{\{r=1\}}\tau^{(0)}_k\Big),
\qquad r\in\{0,1\},\ k\ge1, \label{eq:durations_spell}
\shortintertext{and cycle durations as}
\tau_k &\coloneqq \tau^{(0)}_k+\tau^{(1)}_k,\qquad k\ge 1 \label{eq:durations_cycle}
\end{align}
\end{subequations}

\noindent with the convention that $\sum_{\varnothing}=0$.  In this work, we assume these durations to be finite almost surely (a.s.). We provide in Appendix ~\ref{subsec:spell_durations_finiteness} a sufficient condition on the sequences $(\mathbf{q}^{(0)}, \mathbf{q}^{(1)})$ for this assumption, along with its proof. Moreover, $(\tau^{(0)}_k)_{k\ge 1}$ and $(\tau^{(1)}_k)_{k\ge 1}$ are two sequences of independent and identically distributed (i.i.d.) random variables, and consequently $(\tau_k)_{k\ge 1}$ are i.i.d.. Details are also given in Appendix~\ref{subsec:spell_durations_iid}. From now on, we denote $\tau^{(0)}$, $\tau^{(1)}$, and $\tau$ the random variables distributed as each of these random sequences.

In definition \ref{def:discrete_binary_markov_chain}, the binary sequence  $\{R_0, R_1, \dots \}$  indicates either dry or wet status while the integer-valued  random sequence $\{D_0, D_1, \dots \}$ keeps track of the duration of each dry or wet spell. This is an extension of the model in \cite{1984_stern_model_fitting_daily_rainfall} referred to later as hybrid-order Markov chain \citep{1999_wilks_weather_generation_game_review} in which the value of $D_k$ was upper bounded to keep the number of parameters finite for estimation: see \cite[Table 1]{1999_wilks_interannual_variability_extremes_sdp}. However, this induced an upper bound to the order of memory of the Markov chain. BMCDs alleviate this limitation by providing a very general framework which encompasses for instance classical two-state first-order Markov chains (for any $r \in \{0,1\}$, set $q^{(r)}_d = q^{(r)}_1$ for all $d \geq 1$), two-state second-order Markov chains (for any $r \in \{0,1\}$, set $q^{(r)}_d = q^{(r)}_2$ for all $d \geq 2$), but also a wide range of other models which we characterize in the next paragraphs.

\subsection{BMCD and alternating renewal chains} \label{equivalence_mc_alternating_renewal}

For an introduction to alternating renewal chains we refer the reader to classical textbooks, for example \citet[Section~2.4]{semi_markov_barbu}. Here, an explicit equivalence between the BMCD representation and alternating renewal chains is established: we first show how the BMCD can be written as an alternating renewal chain in equation~\eqref{eq:alternating_renewal_generated_by_spell_durations}. Then we establish the reverse mapping in equation~\eqref{eq:kozubowski_relation_bmcd_from_alternating_chain}. To our knowledge this link has never been properly described in the literature, despite being used implicitly in \cite[Section~3.8]{resnick1992adventures}.

Switching from a BMCD representation to an alternating renewal chain can be described as follows. Using $(\tau^{(0)}_k)_{k\geq1},(\tau^{(1)}_k)_{k\geq1},(\tau_k)_{k\geq1}$ as defined in equation \eqref{eq:durations}, denote
\begin{equation}\label{eq:alternating_renewal_generated_by_spell_durations}
    T_0 \coloneqq 0,\qquad T_k \coloneqq \sum_{i=1}^k \tau_i,\quad k\ge1.
\end{equation}
By construction, $(T_k)_{k\ge0}$ is an \emph{alternating renewal chain} with down-times distributed as $\tau^{(0)}$ and up-times distributed as $\tau^{(1)}$, in the sense of \citet[definition~2.7]{semi_markov_barbu}.
\noindent In an alternating renewal chain, the renewal counting sequence
$(N_n)_{n\in\mathbb{N}}$, defined for $n\in\mathbb{N}$ by 
\begin{equation}\label{eq:renewal_counting_sequence}
N_n := \max\{k \ge 0 : T_k \le n\} = \sum_{k=0}^n \mathbbm{1}_{(R_k,D_k)=(0,1)} - 1,  
\end{equation}
counts the number of dry-wet cycles that have been completed until time $n$. The $-1$ term is added to avoid counting time $T_0=0$ as a renewal time and so that $N_0=0$.

Starting from an alternating renewal chain, an explicit BMCD representation can be recovered. Assume that there are two random variables, $\tau^{(0)}, \tau^{(1)}$, which generate an alternating renewal chain as described in equation \eqref{eq:alternating_renewal_generated_by_spell_durations}. Then a BMCD with spell duration distributions matching the distributions $\tau^{(r)}$ can be built by setting
\begin{equation} \label{eq:kozubowski_relation_bmcd_from_alternating_chain}
    q_d^{(r)} \;=\; 
\begin{cases}
\mathbb{P}(\tau^{(r)} = d \mid \tau^{(r)} \ge d), & \hbox{if} \ \mathbb{P}(\tau^{(r)} \ge d) > 0, \\
1, & \text{otherwise},
\end{cases},\quad \forall d \geq 1.
\end{equation}
This result can be recovered from a direct application of proposition~2.1 from \citet{2025_kozubowski_waiting_time_discrete} to the discrete random variable $\tau^{(0)}_1$. The very same argument applies to $\tau^{(1)}_1$ and the time-shifted Markov chain $(R_{n+\tau^{(0)}_1},D_{n+\tau^{(0)}_1})_{n \geq 0}$.

The equivalence between the BMCD and the alternating renewal representation is illustrated in Fig.~\ref{fig:renewal_process_formalism} using a 15-day synthetic rainfall-occurrence series. For the BMCD representation, the orange and blue bars show the value of the binary indicator $R_n$, and the double arrows indicate the spell durations.  From the renewal chain point of view, $(T_0,T_1 \dots)$ are the renewal times, and $(N_0,N_1 \dots)$ indicates how the renewal counting sequence increments by one at each $T_k$.
\begin{figure}[!htbp]
    \centering
    \includegraphics[width=0.75\linewidth]{ 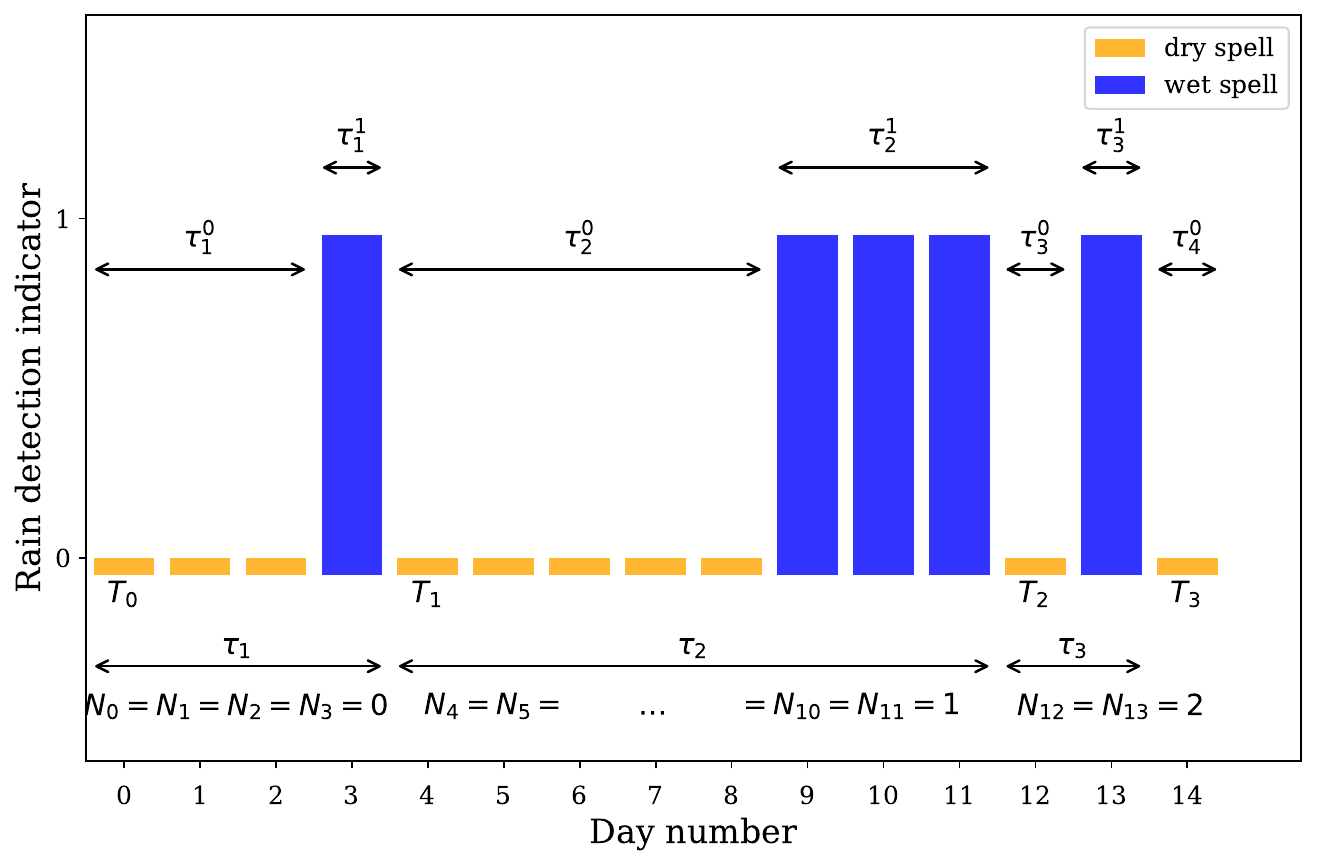}
    \caption{Synthetic rainfall-occurrence series. $R_n\in\{0,1\}$ (definition~\ref{def:discrete_binary_markov_chain}) indicates the dry (in orange) and wet (in blue) states. $(\tau_k^{(0)})_{k \in \mathbb{N}}$, $(\tau_k^{(1)})_{k \in \mathbb{N}}$, and $(\tau_k)_{k \in \mathbb{N}}$ (equations~\eqref{eq:durations_spell}, \eqref{eq:durations_cycle}) (double arrows) indicate spell and cycle duration. $(T_k)_{k \in \mathbb{N}}$, and $(N_n)_{n \in \mathbb{N}}$ (equations~\eqref{eq:alternating_renewal_generated_by_spell_durations}, \eqref{eq:renewal_counting_sequence}) indicate respectively the renewal times and the renewal counting sequence.}
    \label{fig:renewal_process_formalism}
\end{figure}

\subsection{Asymptotic properties} \label{subsec:asymptotic_properties}

Standard asymptotic properties from renewal theory are recalled in Appendix~\ref{subsec:preliminary_results}. Those results are used to prove the following propositions: an asymptotic property from renewal theory is adapted to the BMCD representation, along with an asymptotic normality extension. Then, a central limit theorem is derived.

\begin{proposition}\label{prop:reward_convergence}
Let $(R_n,D_n)_{n\in\mathbb{N}}$ be a BMCD. We suppose that the related cycle duration $\tau$ has a finite first-order moment. Consider a positive function $w:\{0,1\}\times\mathbb{N}^*\to\mathbb{R_+}$ and suppose $\mathbb{E}\!\big[\sum_{k=0}^{\tau-1} w(R_k,D_k)\big] \coloneqq \rho < +\infty$. Then,
\[
\lim_{n\to\infty}\frac{1}{n}\sum_{k=0}^{n} w(R_k,D_k)
= \frac{\rho}{\mathbb{E}[\tau]}
\quad\text{a.s.}
\]

\noindent If $\tau$ has finite second moment, i.e., if $\nu^2 \coloneqq\mathrm{Var} \big( \sum_{k=0}^{\tau-1} w(R_k, D_k) \big) <+\infty$, then:
$$
\frac{\sum_{k=0}^{n} w(R_k, D_k) -N_ n\rho}{\sqrt{n \nu^2 / \mathbb{E}[\tau]}} \;\;\xrightarrow[n \to \infty]{\;\;d\;\;}\;\; \mathcal{N}(0,1).
$$
\end{proposition}

\noindent The proof is given in Appendix~\ref{proof-reward_convergence}. This asymptotic property provides a closed-form expression for a quantity that is typically only investigated through numerical experiments in the literature (see, for example, \citet{vu2018evaluation}). To illustrate its practical utility, we consider the asymptotic proportion of time spent in long dry spells, a quantity of considerable importance in crop production \citep{jarrett2023dry}. Example~\ref{ex:long_dry_spell} yields a closed-form expression for this proportion, which may be numerically approximated to arbitrary precision using the specification of Section~\ref{subsec:spell_duration_distribution_specification}, using the bounds derived in Appendix~\ref{subsec:approx_proportion_time_long_dry_spell}.
\begin{example}[Proportion of time in long dry spells]\label{ex:long_dry_spell}

For $d\in\mathbb{N}^*$, set $w(R_k,D_k)=\mathbbm{1}_{\{R_k=0,\; D_k\ge d\}}$. Then,
\[
\lim_{n\to\infty}\frac{1}{n}\sum_{k=0}^{n} \mathbbm{1}_{\{R_k=0,\; D_k\ge d\}}
= \frac{\mathbb{E}\bigl[\max(0,\tau^{(0)}-d)\bigr]}{\mathbb{E}[\tau]}
\quad\text{a.s.}
\]
\end{example}

We now state a central limit theorem which will be used later for goodness-of-fit. For $r \in \{0,1\}$, denote 
\begin{equation}\label{eq:nb_spell_longer_than_d}
N_{n}^{(r)}(d)\coloneqq\sum_{k=0}^{n}\mathbbm{1}(R_k=r,D_k=d)
\end{equation}
the number of $r$-spells longer than $d$ days.

\begin{proposition}\label{prop:cltintermediate}
Let $(R_{k},D_{k})_{1\le k\le n}$ be a BMCD with parameters
$\mathbf{q}^{(r)}$. Assume that the cycle duration $\tau$ has a second-order moment. Then, for any integer $d_{\max} > 2$,
\[
\sqrt{\frac{n}{\mathbb{E}[\tau]}}
\Bigg(
\frac{N_{n}^{(r)}(d)}{N_n}
-\overline F_{\tau^{(r)}}(d-1)
\Bigg)_{2\le d\le d_{\max}}
\;\xrightarrow[n \to \infty]{\;d\;}
\mathcal{N}_{d_{\max}-1}(\mathbf{0},\Sigma),
\]
where the covariance matrix $\Sigma=\big(\Sigma_{i,j}\big)_{1\le i,j\le d_{\max}-1}$
is defined by
\[
\Sigma_{i,j}
=
\overline F_{\tau^{(r)}}\!\big(\max(i,j)\big)
-\overline F_{\tau^{(r)}}(i)\,\overline F_{\tau^{(r)}}(j),
\qquad
1\le i,j\le d_{\max}-1.
\]
Which gives,
\[
\Sigma
=
\begin{pmatrix}
\overline F_{\tau^{(r)}}(1)\bigl(1-\overline F_{\tau^{(r)}}(1)\bigr)
& \cdots
& \overline F_{\tau^{(r)}}(d_{\max}-1)\bigl(1-\overline F_{\tau^{(r)}}(1)\bigr)
\\
\vdots & \ddots & \vdots
\\
\overline F_{\tau^{(r)}}(d_{\max}-1)\bigl(1-\overline F_{\tau^{(r)}}(1)\bigr)
& \cdots
& \overline F_{\tau^{(r)}}(d_{\max}-1)\bigl(1-\overline F_{\tau^{(r)}}(d_{\max}-1)\bigr)
\end{pmatrix}.
\]
\end{proposition}

\noindent The proof is given in Appendix~\ref{proof:prop_cltintermediate}. This proposition will be used to derive a goodness-of-fit test for the validation of the model in Section~\ref{subsec:validation_tools}. The integer $d_{\max}$ serves as a hyperparameter of this statistical test.

\subsection{Specification of spell duration distributions}\label{subsec:spell_duration_distribution_specification}

Estimating $\mathbf{q}^{(r)}$, $r\in\{0,1\}$, raises the issue already mentioned in \citet{1984_stern_model_fitting_daily_rainfall}: increasing the Markov memory inflates the number of parameters. To avoid this parameter explosion, we model spell durations through low-dimensional parametric families, and then recover the sequences $\mathbf{q}^{(r)}$ from equation~\eqref{eq:kozubowski_relation_bmcd_from_alternating_chain}. A wide range of models has been proposed for spell durations, including variations of geometric distributions or negative binomial distributions  \citep{1977_buishand_stochastic_modelling_daily_rainfall, 1991_racsko_serial_local_stochastic_weather_models, 2010_deni_probability_models_dry_wet_spells}. In our data, dry spell durations display markedly different upper-tail behaviours across regions, from light to heavy tails, whereas wet spell durations are generally exponentially-tailed. To handle this variability within a unified framework, the distribution for each spell duration is constructed as a particular case of the eGPD class \citep{2016_naveau_modelling_jointly_rainfall_intensities, 2026_Naveau}, which extends GP modelling from the tail to the whole support. Further details on this class and the derivation for each spell duration are given in Appendix~\ref{subsec:general_class_egpd}, and the resulting specifications are presented below.

For dry spells, the specification of the general eGPD class (detailed in Appendix~\ref{subsec:general_class_egpd}) has an additional mass at duration $1$. We call the distribution a hurdle discretised eGPD (hdeGPD).
\begin{equation}\label{eq:tau0_pmf}
\mathbb{P}_{f_1,\kappa,\sigma,\xi}\!\big(\tau^{(0)}=d\big)=
\begin{cases}
f_1, & d=1,\\
(1-f_1)\,\Big[F_{\kappa,\sigma,\xi}(d-1)-F_{\kappa,\sigma,\xi}(d-2)\Big], & d\ge 2,
\end{cases}
\end{equation}
where $F_{\kappa,\sigma,\xi}$ denotes the type-1 eGPD cdf
\begin{equation}\label{eq:type-1-eGPD}
F_{\kappa,\sigma,\xi}(z)=
\begin{cases}
\bigl(1-(1+\xi z/\sigma)^{-1/\xi}\bigr)^\kappa,
  & (\xi>0, z\ge0)\ \text{or}\ (\xi<0,0<z<-\sigma/\xi),\\[2pt]
\bigl(1-\exp(-z/\sigma)\bigr)^\kappa,
  & \xi=0, z\ge0,\\[2pt]
0, & \text{otherwise.}
\end{cases}
\end{equation}
The parameter $\xi$ drives the upper tail behaviour (similarly the so-called parameter $\kappa>0$ drives the lower Pareto tail). In particular, $\xi>0$ yields heavy (Pareto-type) tails, $\xi=0$ subexponential tails, and $\xi<0$ a bounded upper tail.

For wet spells, we use another specification within the same general eGPD class (detailed in Appendix~\ref{subsec:general_class_egpd}), with a constraint $\xi=0$ enforcing a subexponential distribution. The final result reduces to a mixture of two geometric distributions:
\begin{equation}\label{eq:mixgeom_pmf_wet}
\mathbb{P}_{\pi,p_1,p_2}\!\big(\tau^{(1)}=d\big)
=\pi\, p_1(1-p_1)^{d-1}+(1-\pi)\, p_2(1-p_2)^{d-1},
\qquad d\in\mathbb{N}^*.
\end{equation}
This distribution is slightly more complex than the classical geometric distribution \citep{1981_richardson_stochastic_simulation_weather, 1991_racsko_serial_local_stochastic_weather_models, 1998_richardson_semenov_comparison_wgen_larswg,2014_chen_review_stochastic_generation_precipitation} but substantially more flexible \citep{2010_deni_probability_models_dry_wet_spells}, and provides an excellent fit across stations and seasons in southern Europe (see in Section~\ref{sec:application_europe_data}). Also, it is  a meteorologically supported choice, as rainfall events may arise from two different mechanisms, such as convective and stratiform regimes, even if those are two very general categories.

The mixture model includes several important limiting cases. When $p_1\to 1$, the first component concentrates on duration $1$, so $\pi$ plays a role similar to $f_1$ in the dry spell model, although the remaining tail is geometric rather than Pareto. When $p_1\to p_2$, or $\pi\to0$, or $\pi\to1$, the model reduces to a single geometric and is not identifiable any more. This single geometric limiting case does not affect model performance but means that a simple geometric model would have been sufficient.

\subsection{Estimation of the parameters}\label{subsec:estimation_parameters_methodology}
The parameters of the distribution of $\tau^{(0)}$ as given in equation~\eqref{eq:tau0_pmf} are $(f_1,\kappa,\sigma,\xi)$. We estimate separately the mass at $d=1$ and the parameters of the shifted distribution when $d \ge 2$. The probability mass at $d=1$  is estimated by the empirical frequency $$\widehat f_1=\frac{1}{n}\sum_{k=1}^n \mathbbm{1}\{\tau^{(0)}_k=1\}.$$ For the remaining parameters $(\kappa,\sigma,\xi)$, we apply the Probability Weighted Moments (PWM) method detailed in \citet[Section~3.1]{2016_naveau_modelling_jointly_rainfall_intensities}, which is a variation of the method-of-moments based on the moments 
$$
\mathbb{E}\!\left[X\,\overline F^s_{\kappa,\sigma,\xi}(X)\right], \quad s=0,1,2$$ 
for the random variable $X$. PWM is commonly used in extreme value analysis and in hydrology \citep{1987_HoskingWallis_GPD}. Specifically, we apply this method to the shifted sample $\{\tau^{(0)}_k-2:\ \tau^{(0)}_k\ge 2\}$. Then the the resulting system is numerically solved just as in the package {\tt mev} \citep{mev2025}. As demonstrated in the case study of \citet[Section~4]{2016_naveau_modelling_jointly_rainfall_intensities}, the PWM estimator exhibits robustness to discretisation arising from instrumental precision.

Regarding the wet spell durations $\tau^{(1)}$, the parameters of the distribution in~\eqref{eq:mixgeom_pmf_wet} are estimated by maximum likelihood using the EM algorithm (details are provided in Appendix~\ref{subsec:em_algorithm_mixture_geometric}). Identifiability is enforced by imposing $p_1 > p_2$.

\subsection{Validation of the model}\label{subsec:validation_tools}

To assess the validity of the BMCD and of our inference scheme, we use autocorrelations to test for independence of successive spell durations (detailed in Appendix~\ref{subsec:autocorrelation_bivariate_spell_duration}), Q-Q plots for validating the fits on the marginal distribution (detailed in Appendix~\ref{subsec:simulation_based_qqplots_explanation}), and a goodness-of-fit statistic specifically designed for exit state probabilities, detailed in the following.

A distinctive feature of the BMCD is that the exit probability of a given state depends on the current duration $d$ through $q^{(r)}_d$. A constant $q_d^{(r)}= q$ leads to a memoryless two-state first-order Markov chain. Deviations from this constant indicate that the probability of exit depends on how long the sequence has remained in the current state.
For dry spells ($r=0$), there is often a decreasing part in the curve of $q_d^{(0)}$ as a function of $d$. This decreasing pattern indicates persistence: the longer the dry spell, the less likely it is to end. This persistence is often described in the literature \citep{wilks2011statistical,lopez2015markovian} but is typically limited to a given order within a Markov model. 
Since there is no such limit to the BMCD, a visual diagnostic tool focusing on estimates of $q^{(r)}_d$ can be designed. Once a parametric distribution for $\tau^{(r)}$ is fitted, yielding parameter estimates $\hat\theta$, an immediate evaluation of these probabilities can be obtained, using equation~\eqref{eq:kozubowski_relation_bmcd_from_alternating_chain}:
\begin{equation}\label{eq:link_tau_exit_proba}
\widehat q^{(r)}_{d,\hat \theta}
\;=\;
\frac{ \mathbb{P}_{\hat\theta}(\tau^{(r)}=d)}{ \mathbb{P}_{\hat\theta}(\tau^{(r)}\ge d)},
\end{equation}
where $\mathbb{P}_{\hat\theta}$ denotes probabilities under the fitted distribution.

\noindent We compare this curve, derived from the distribution \eqref{eq:tau0_pmf} with parameters $\hat \theta$, to a non-parametric empirical estimate computed directly from the recorded spell durations. Using $N_{n}^{(r)}(d)$ from equation \ref{eq:nb_spell_longer_than_d}, the empirical estimator is
\begin{equation}\label{eq:exit_proba_empirical_estimator}
\widehat q^{(r)}_{d,\mathrm{emp}}
\;=\;
\frac{N_{n}^{(r)}(d) - N_{n}^{(r)}(d+1)}{N_{n}^{(r)}(d)}.
\end{equation}

\noindent To visualise the sampling variability of the exit-probability estimate at each duration $d$, we report pointwise bands based on a binomial model for the number of exits at day $d$,  $\sqrt{\frac{\widehat q^{(r)}_{d,\hat \theta}\bigl(1-\widehat q^{(r)}_{d,\hat \theta}\bigr)}{N_{n}^{(r)}(d)}}.$ 
This diagnostic is therefore useful not only to validate the relevance of a duration-dependent model such as the BMCD, but also to check that the fitted spell-duration model exit probabilities $\widehat q_{d,\hat \theta}^{(r)}$ matches reasonably the empirical exit probabilities $\widehat q_{d,\mathrm{emp}}^{(r)}$. To go even further in that comparison, we now describe a goodness-of-fit test.

Let $(R_n,D_n)_{n\ge 0} \in \{0,1\} \times \mathbb{N}^*$ be a sequence of observations. If one assumes the sequence is generated by a BMCD, then for each regime $r\in\{0,1\}$, the model is characterized by exit probabilities $\mathbf q^{(r)}$, and one can recover $\mathbf{\hat{q}_{\hat \theta}^{(r)}} := \bigl( \hat q_{d,\hat \theta}^{(r)}\bigr)_{d\ge 1}$ from \eqref{eq:link_tau_exit_proba} and $\mathbf{\hat{q}_{\mathrm{emp}}^{(r)}} := \bigl(\hat q_{d,\mathrm{emp}}^{(r)}\bigr)_{d\ge 1}$ from \eqref{eq:exit_proba_empirical_estimator}. Thus, define the following testing hypothesis:
$$H_0^{(r)}:\ \exists\,\theta_0\in\Theta\ \text{such that}\ \mathbf{q}^{(r)}=\mathbf{q}_{\theta_0}^{(r)} \qquad\text{vs.}\qquad H_1^{(r)}:\ \forall\,\theta\in\Theta,\quad \mathbf{q}^{(r)}\neq\mathbf{q}_{\theta}^{(r)}.$$
\begin{proposition}[Chi-squared goodness-of-fit test]\label{prop:gof_q}
Consider $r\in\{0,1\}$ and an integer $d_{\max}\ge 2$. Suppose $H_0^{(r)}$ and assume the conditions of
proposition~\ref{prop:cltintermediate}. Throughout this proposition, $r$ is fixed and we write $F_{\theta}$ for the cdf of $\tau^{(r)}$ under the BMCD with exit probabilities $\mathbf{q}_{\theta}^{(r)}$. Denote $\Sigma_{\theta_0}$ the covariance matrix defined in
proposition~\ref{prop:cltintermediate} with $\theta_0$ defined in $H_0^{(r)}$. For $i \text{ and } j=1,\dots,d_{\max}-1$, define the matrix
$\mathbf{T}_{\theta_0}$ by
\[
(\mathbf T_{\theta_0})_{i,j}
=
-\mathbbm 1_{\{i=j=1\}}
\;+\;
\mathbbm 1_{\{i\ge 2\}}
\left[
\frac{\overline F_{\theta_0}(i)}{\overline F_{\theta_0}(i-1)^2}\,\mathbbm 1_{\{j=i-1\}}
-\frac{1}{\overline F_{\theta_0}(i-1)}\,\mathbbm 1_{\{j=i\}}
\right].
\]

If $\mathbf{T}_{\theta_0}\Sigma_{\theta_0}^2\mathbf{T}^{\mathsf T}_{\theta_0}$ is nonsingular, then
\[
\mathcal{Q}_{N_n} :=N_n\,
\mathbf{\Delta}^{\mathsf T}
\left(\mathbf{T}_{\theta_0}\Sigma_{\theta_0}\mathbf{T}^{\mathsf T}_{\theta_0}\right)^{-1}
\mathbf{\Delta}
\;\stackrel{d}{\longrightarrow}\;
\chi^2_{d_{\max}-1},
\]
with
\[
\mathbf{\Delta} := \bigl(\widehat {\mathbf{q}}_{\mathrm{emp}|1:d_{\max}-1}^{(r)}-\widehat {\mathbf{q}}_{\theta_0 |1:d_{\max}-1}^{(r)}\bigr).
\]
\end{proposition}

\noindent The proof of Proposition \ref{prop:gof_q} is given in Appendix~\ref{proof-gof}. We also show the validity of the goodness-of-fit test on finite samples simulated under $H_0^{(r)}$ in Appendix~\ref{subsec:gof_simulated_data}. In practice, $\theta_0$ is unknown and is replaced by an estimator $\hat\theta$: in our case we use the one described in Section \ref{subsec:estimation_parameters_methodology}. This gives the elements for a statistical test. Set $\alpha\in(0,1)$ and $d_{\max}\ge2$. Compute $\mathbf\Delta_{\hat \theta}:=\widehat {\mathbf{q}}_{\mathrm{emp}|1:d_{\max}-1}^{(r)}-\widehat {\mathbf{q}}_{\hat \theta|1:d_{\max}-1}^{(r)}$, and the plug-in matrices $\Sigma_{\hat\theta}^2$ and $\mathbf T_{\hat\theta}$. Define
\[
{\mathcal{Q}}_{N_n,\; \hat \theta}:=N_n\,\mathbf\Delta_{\hat \theta}^{\mathsf T}
\left(\mathbf{T}_{\hat \theta}\Sigma_{\hat \theta}\mathbf{T}^{\mathsf T}_{\hat \theta}\right)^{-1}
\mathbf\Delta_{\hat \theta}.
\]
Under $H_0^{(r)}$, $\mathcal{Q}_{N_n} \stackrel{d}{\longrightarrow} \chi^2_{d_{\max}-1}$; thus reject $H_0^{(r)}$ at level $\alpha$ if
\[
\mathcal{Q}_{N_n,\; \hat \theta}>\chi^2_{d_{\max}-1,\,1-\alpha},
\]
provided $\mathbf{T}_{\hat \theta}\Sigma_{\hat \theta}^2\mathbf{T}^{\mathsf T}_{\hat \theta}$ is nonsingular.

\section{Application to southern European dry and wet spells}\label{sec:application_europe_data}

We fit our model to southern Europe precipitation data from the European Climate Assessment \& Dataset (ECA\&D). This dataset stems from an operational effort to collect, quality-check, and disseminate daily meteorological observations across Europe and the Mediterranean. The database is actively maintained and concatenates time series provided by national meteorological services and partners. It provides station-based daily totals of precipitation. Quality-control is applied to each observation including physical-range and repetitiveness checks (repeated identical values over multiple days), among others. Homogeneity is assessed through several standard checks. More details can be found in the Algorithm Theoretical Basis Document \url{https://www.ecad.eu/documents/atbd.pdf}. Europe is well covered by ECA\&D, with tens of thousands of stations today. We use a spatially uniform subset of around 200 stations covering the study area (latitude below 45°N). For missing precipitation records, sequences of four or more consecutive missing days are removed entirely, while gaps of three days or fewer are filled by linear interpolation. We deal separately with the spells from the four standard climatological seasons: winter (December January February), spring (March April May), summer (June July August), autumn (September October November). The spell is assigned to the season of its start date, which is standard practice \citep{1977_buishand_stochastic_modelling_daily_rainfall}. To ensure fully recorded spells, we drop the first and the last spell in each continuous segment of recorded days: this avoids partial spells induced by the arbitrary start/end of the observation windows. We only consider records from 1945 onwards, and exclude stations with fewer than 30 cumulative years of data. We classify a day as wet if the recorded daily precipitation accumulation is higher than 0.6 mm, and dry otherwise. This threshold methodology to define a wet spell is standard practice \citep{domroes1993statistical} to mitigate low-intensity measurement errors. We treat each season and each location separately. Temporal stationarity is checked in the code repository.

\subsection{Estimation of the parameters}\label{subsec:param_estimation_ecad}

Figs.~\ref{fig:hist_dry_spell_parameters} and \ref{fig:hist_wet_spell_parameters} display histograms of the estimated parameters for the dry spell distribution \eqref{eq:tau0_pmf} and the wet spell distribution \eqref{eq:mixgeom_pmf_wet}, respectively. Within each figure, columns represent individual parameter estimates, arranged from left to right as $\hat f_1$, $\hat \xi$, $\hat \sigma$, and $\hat \kappa$ for the dry spell distribution, and $\hat \pi$, $\hat p_1$, and $\hat p_2$ for the wet spell distribution. Rows represent the seasons, arranged from top to bottom as spring, summer, autumn, and winter. 

The estimates $\hat{\xi}$ (second column of 
Fig.~\ref{fig:hist_dry_spell_parameters}) range from approximately $-0.5$ (upper-bounded tail fitted distribution), to approximately $0.5$ (heavy-tailed fitted distribution), with a modal value near $0$ (exponentially-tailed fitted distribution). These variations across stations reflect the considerable heterogeneity in dry spell dynamics present in the dataset, and support the use of the hdeGPD specification, which accommodates flexible upper-tail behaviour. Most cases of high positive $\hat \xi$ (heavy-tailed dry spell duration) are reported in spring, a feature already alluded to in Fig.~\ref{fig:survival_func_sample_data_intro} analysis. The estimates $\hat f_1$, $\hat \sigma$ and $\hat \kappa$ (first, third and fourth columns of 
Fig.~\ref{fig:hist_dry_spell_parameters})  exhibit broadly similar distributions across seasons, with the notable exception of summer. Specifically, $\hat f_1$, the probability of having a $1$-day long dry spell, is slightly lower in summer (from around $0.25$ to around $0.15$), and $\hat \sigma$, the scale factor, takes considerably larger values (from around $10$ to values between $50$-$100$), also in summer. This behaviour is consistent with the tendency for long dry spells to be common in summer, unlike for other seasons for which they are extreme events.

The estimates $\hat{p}_1$ and $\hat{p}_2$  (second and third columns of Fig.~\ref{fig:hist_wet_spell_parameters}) exhibit broadly similar distributions across spring, autumn, and winter, indicating a consistent wet spell structure across these three seasons. The notable exception is summer, for which both $\hat{p}_1$ and $\hat{p}_2$ take substantially higher values. This behaviour is consistent with the tendency for summer wet spells to be markedly shorter than those occurring in other seasons. The estimates $\hat{\pi}$ (first column of Fig.~\ref{fig:hist_wet_spell_parameters}) ranges on the full unit interval, with a minor proportion of stations whose estimates are close to $0$ or $1$. In these boundary cases, the mixture of geometric distributions \eqref{eq:mixgeom_pmf_wet} effectively degenerates towards a single geometric component, so that the model retains the flexibility of a two-component mixture whilst remaining consistent with a simpler geometric specification wherever the data so dictate. A similar degeneracy arises when $\hat{p}_1 \simeq \hat{p}_2$, which cannot be noted from those histograms, but may arise in some cases (such as in Palermo-spring and Palermo-autumn as shown in the legend of Fig.~\ref{fig:histograms_fitted_geom_mix}).These observations confirm that the two-component mixture is fully exploited by the data: the additional complexity is retained only where the empirical wet spell durations genuinely exhibit two-scale behaviour, and is otherwise reduced to the simpler limiting case. 

\begin{figure}[!htbp]
    \centering
    \includegraphics[trim=0 14 0 0,clip,width=0.9\linewidth]{ 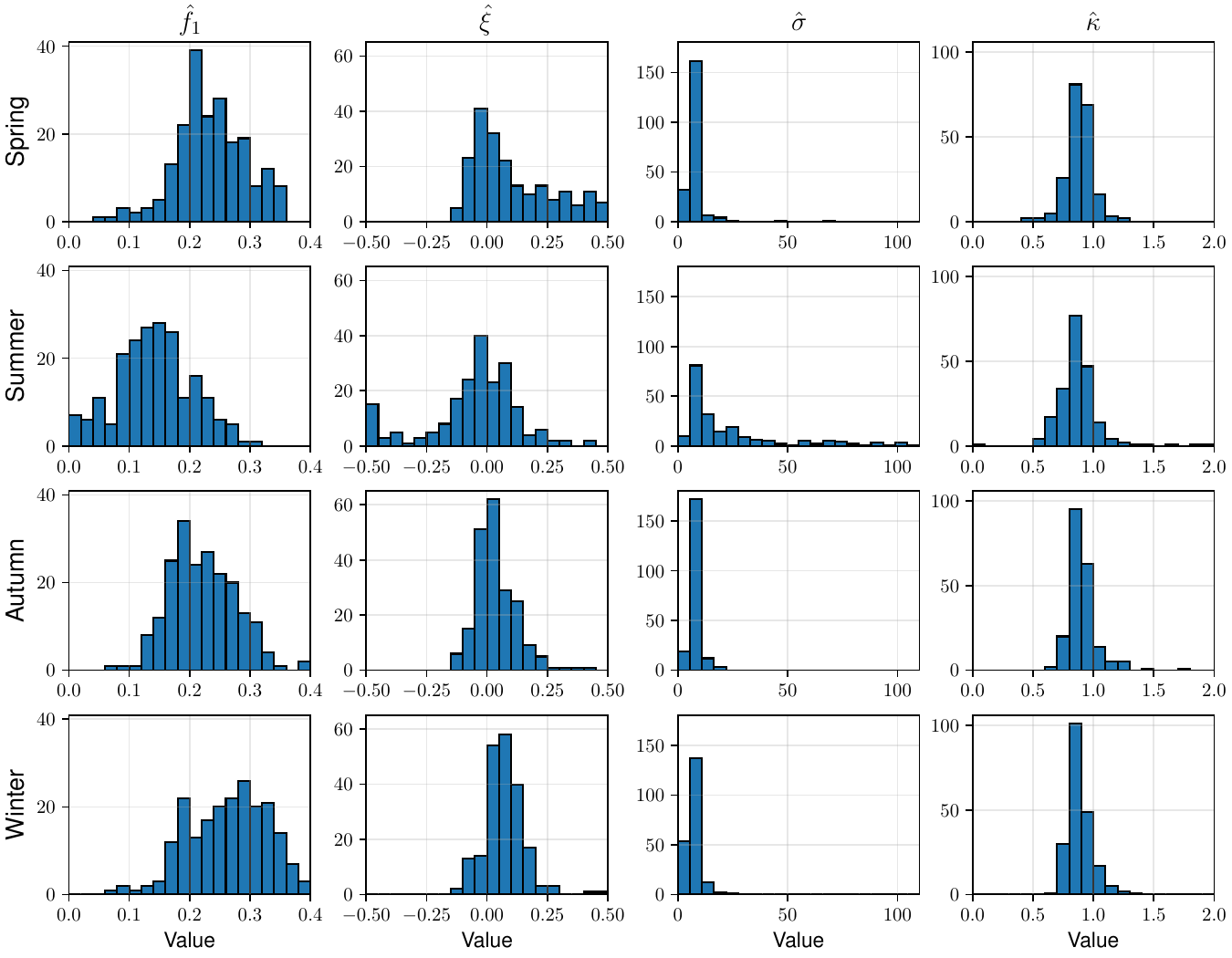}
    \caption{Histogram of each fitted parameters (from left to right: $\hat f_1$, $\hat \xi$, $\hat \sigma$, $\hat \kappa$) of dry spells $\tau^{(0)}$ distribution as specified in equation~\eqref{eq:tau0_pmf}, for every season (from top to bottom: spring, summer, autumn, winter).}
    \label{fig:hist_dry_spell_parameters}
\end{figure}

\begin{figure}[!htbp]
    \centering
    \includegraphics[trim=0 14 0 0,clip,width=0.8\linewidth]{ 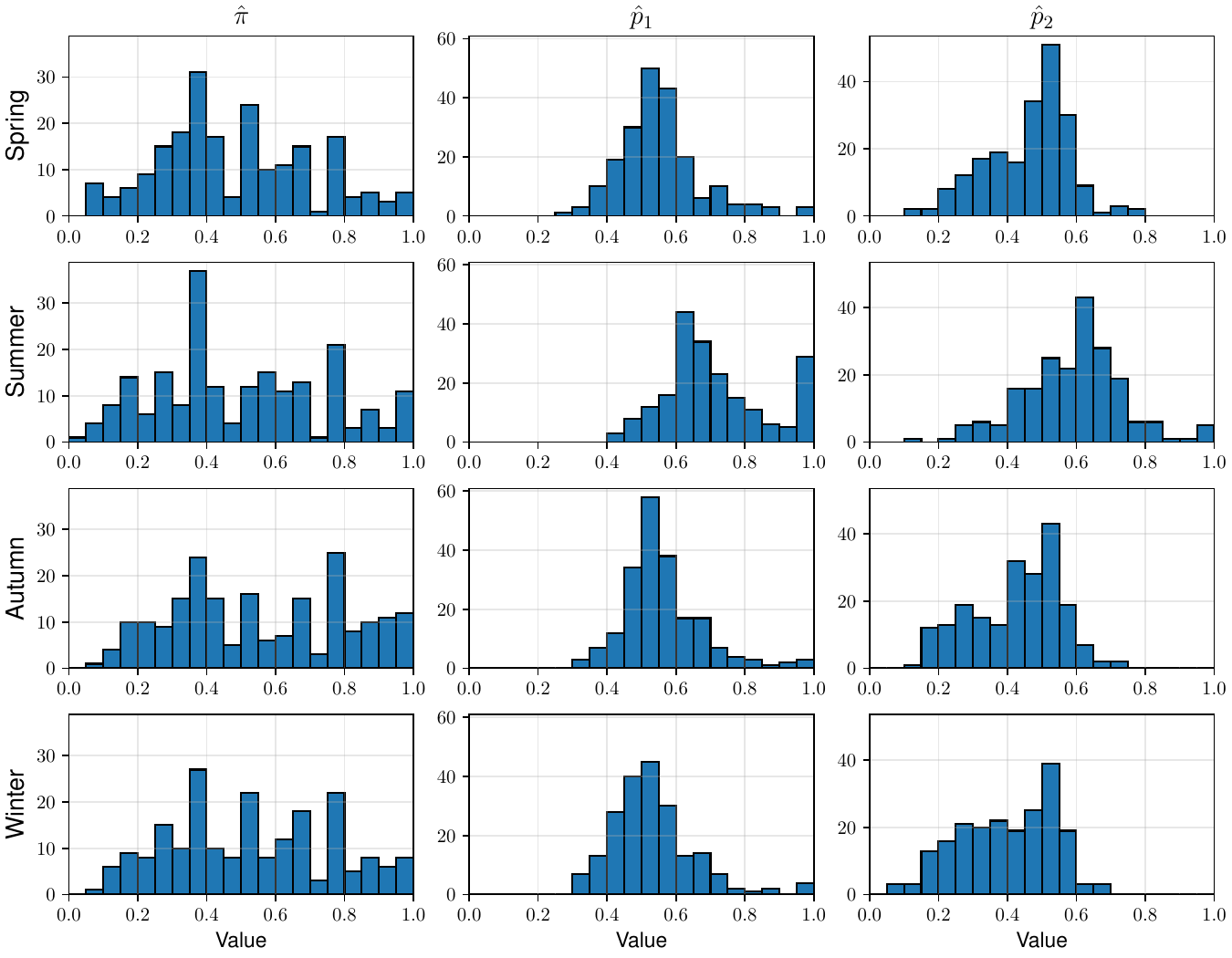}
    \caption{Histogram of each fitted parameters (from left to right: $\hat \pi$, $\hat p_1$, $\hat p_2$) of wet spells $\tau^{(1)}$ distribution as specified in equation~\eqref{eq:mixgeom_pmf_wet}, for every season (from top to bottom: spring, summer, autumn, winter).}
    \label{fig:hist_wet_spell_parameters}
\end{figure}

\subsection{Validation of the model}\label{subsec:model_validation_and_robustness}

In the following we run a series of diagnostics to assess the validity and robustness of our model. We begin by assessing the independence of consecutive spell durations, a key requirement of the BMCD stated right after definition~\ref{def:discrete_binary_markov_chain}. We use an autocorrelation function of the bivariate spell-duration process, introduced in Appendix \ref{subsec:autocorrelation_bivariate_spell_duration}. Fig.~\ref{fig:acf_plots} displays, for every season in Palermo (Italy), the estimated autocorrelation for spell durations (dry-dry (\textcolor{orange}{$\bullet$}), dry-wet ({\tiny $\blacksquare$}), wet-dry ($\blacktriangle$), and wet-wet (\textcolor{blue}{\fontsize{8}{8}\selectfont \scalebox{1.5}[1]{$\blacklozenge$}})) at each lag. A horizontal reference line at 0 highlights the absence of serial dependence, while the grey shaded envelope represents the reference bounds $\pm 2/\sqrt{C_\ell}$ (with $C_\ell$ the number of valid pairs at lag $\ell$). In Palermo-summer (top-right panel), long dry spells results in few spell cycles during a season, so the autocorrelation can only only be estimated at low lags. Shorter dry spells in Palermo-winter (bottom-right panel) result in many short spell-cycles in a season, thus autocorrelation can be estimated for greater lags. For all seasons, autocorrelations remain largely within the reference bounds and close to zero, except for high lags when $C_\ell$ is small and sampling variability is high. These results support the assumption that consecutive spell durations are independent.

\begin{figure}[!htbp]
    \centering
    \includegraphics[trim=0 8 0 5,clip,width=0.75\linewidth]{ 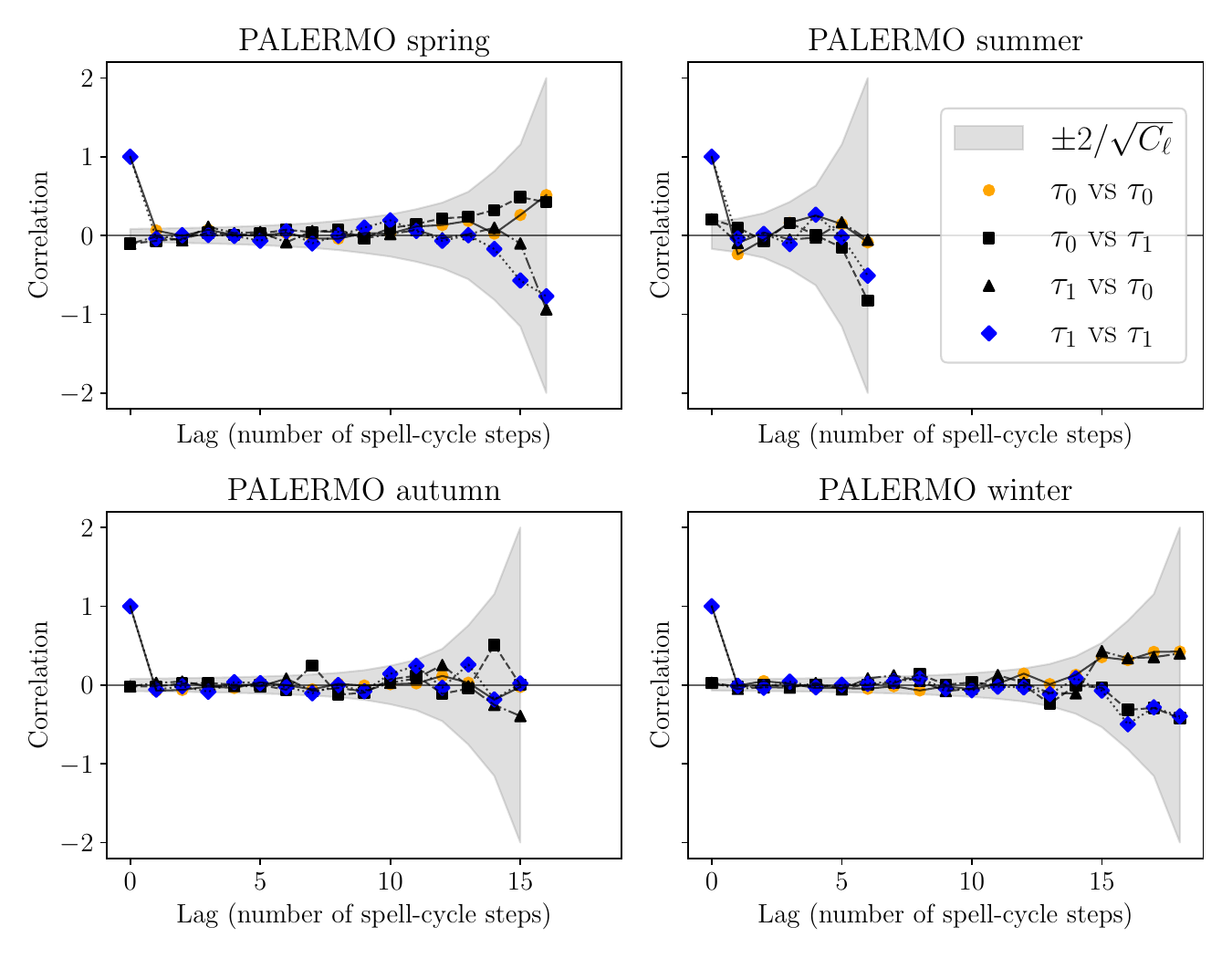}
    \caption{Bivariate autocorrelation of spell durations (detailed in Appendix~\ref{subsec:autocorrelation_bivariate_spell_duration}) in Palermo for every season. The legend (top-right panel) shows dry-dry (\textcolor{orange}{$\bullet$}), dry-wet ({\tiny $\blacksquare$}), wet-dry ($\blacktriangle$), and wet-wet (\textcolor{blue}{\fontsize{8}{8}\selectfont \scalebox{1.5}[1]{$\blacklozenge$}}) autocorrelations symbols for each lag. The grey band shows the reference bounds (also detailed in Appendix~\ref{subsec:autocorrelation_bivariate_spell_duration})).}
    \label{fig:acf_plots}
\end{figure}

Next, the fitted spell-duration probability mass function (pmf) from \eqref{eq:tau0_pmf} for dry spells or \eqref{eq:mixgeom_pmf_wet} for wet spells is compared to its empirical counterpart. Figs.~\ref{fig:histograms_fitted_ext_gpd} and \ref{fig:histograms_fitted_geom_mix} display the empirical histogram (coloured bars) and overlay the fitted probability mass function (black curve). The close alignment of both pmfs indicates a good fit of the bulk of the distribution. Figs.~\ref{fig:qqplot_fitted_dry} and \ref{fig:qqplot_fitted_wet} display the simulation-based Q-Q plots detailed in Appendix~\ref{subsec:simulation_based_qqplots_explanation}. When fitted distributions of dry spell duration have a high value of $\hat \xi$ such as Palermo-spring (top-left panel), with $\hat \xi = 0.32$, the upper quantiles exhibit high variance, reflecting the possibility of extremely long dry spells relative to the bulk. It is worth emphasising that the $300$-day simulated dry spell is an estimate of the quantile of order $0.05$ of the longest dry spell in about $70$ years, and therefore represents a genuinely rare extreme event. It should also be noted that with our season assignment method, a dry spell of this magnitude beginning in spring would span summer and autumn, and possibly extend into winter. In Palermo-summer (top-right panel), the fitted scale parameter is large $\hat \sigma = 28.5$ but the shape parameter is near zero $\hat \xi = 0.02$ producing consistently high quantiles with lower variance. For every season, the points lie close to the identity line and within the bootstrap envelope, showing that simulated quantiles from the fitted model are close to empirical ones. Wet spell simulated quantiles lie close to the identity line and within the bootstrap envelope across all seasons. Unlike dry spells, the upper quantiles exhibit markedly lower variance, reflecting the subexponential tail of the fitted mixed geometric distribution, which makes extreme wet spell durations unlikely to greatly exceed recorded ones.


\begin{figure}[!htbp]
    \centering
        \includegraphics[trim=25 20 25 25,clip,width=0.75\linewidth]{ 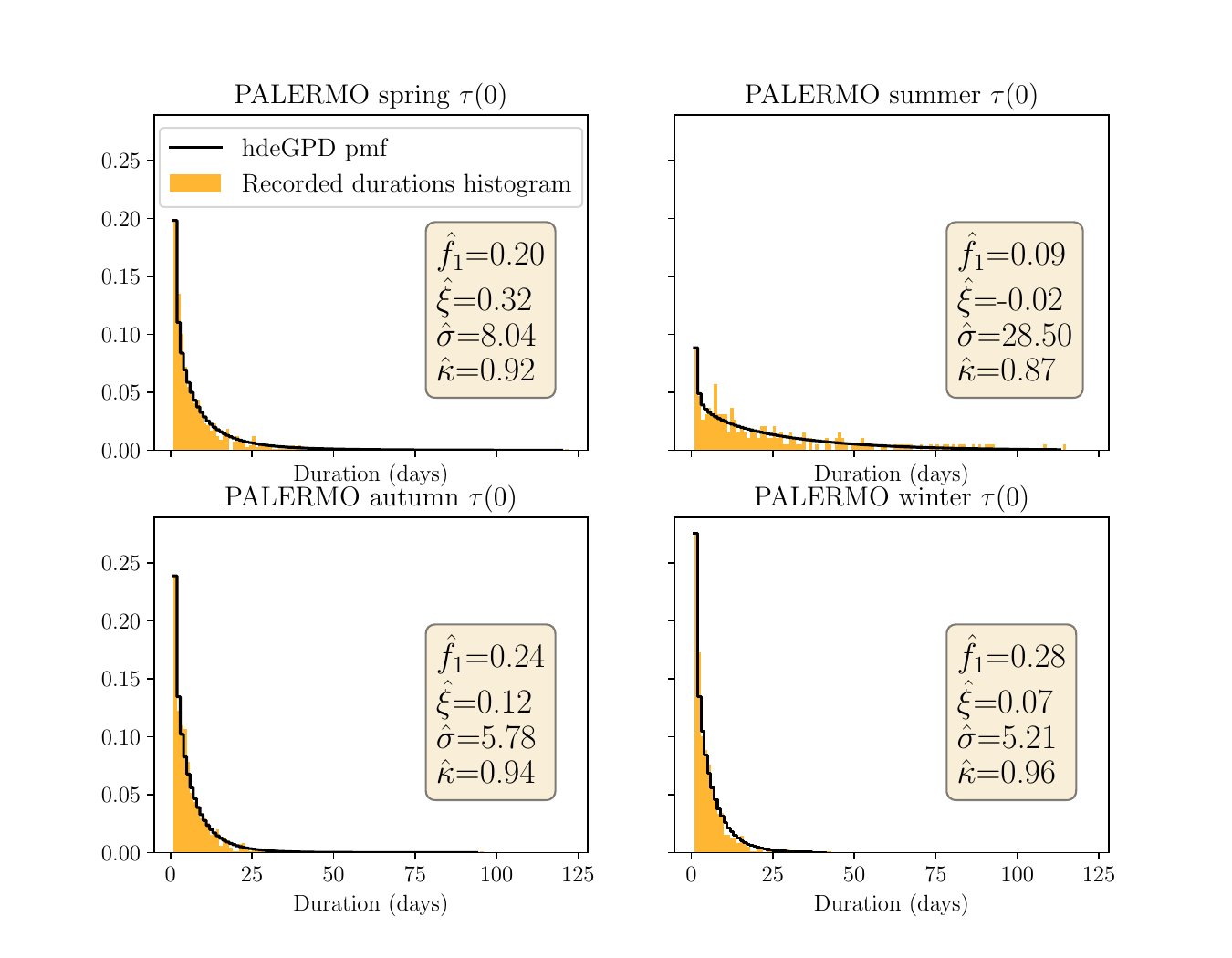}
    \caption{Histogram of recorded dry spell durations $\tau^{(0)}$ in Palermo and for every season. The black curve is the fitted discretised extended-GPD probability mass function \eqref{eq:tau0_pmf}. The fitted parameters are indicated in the legend.}
    \label{fig:histograms_fitted_ext_gpd}
\end{figure}

\begin{figure}[!htbp]
    \centering
        \includegraphics[trim=25 20 25 25,clip,width=0.75\linewidth]{ 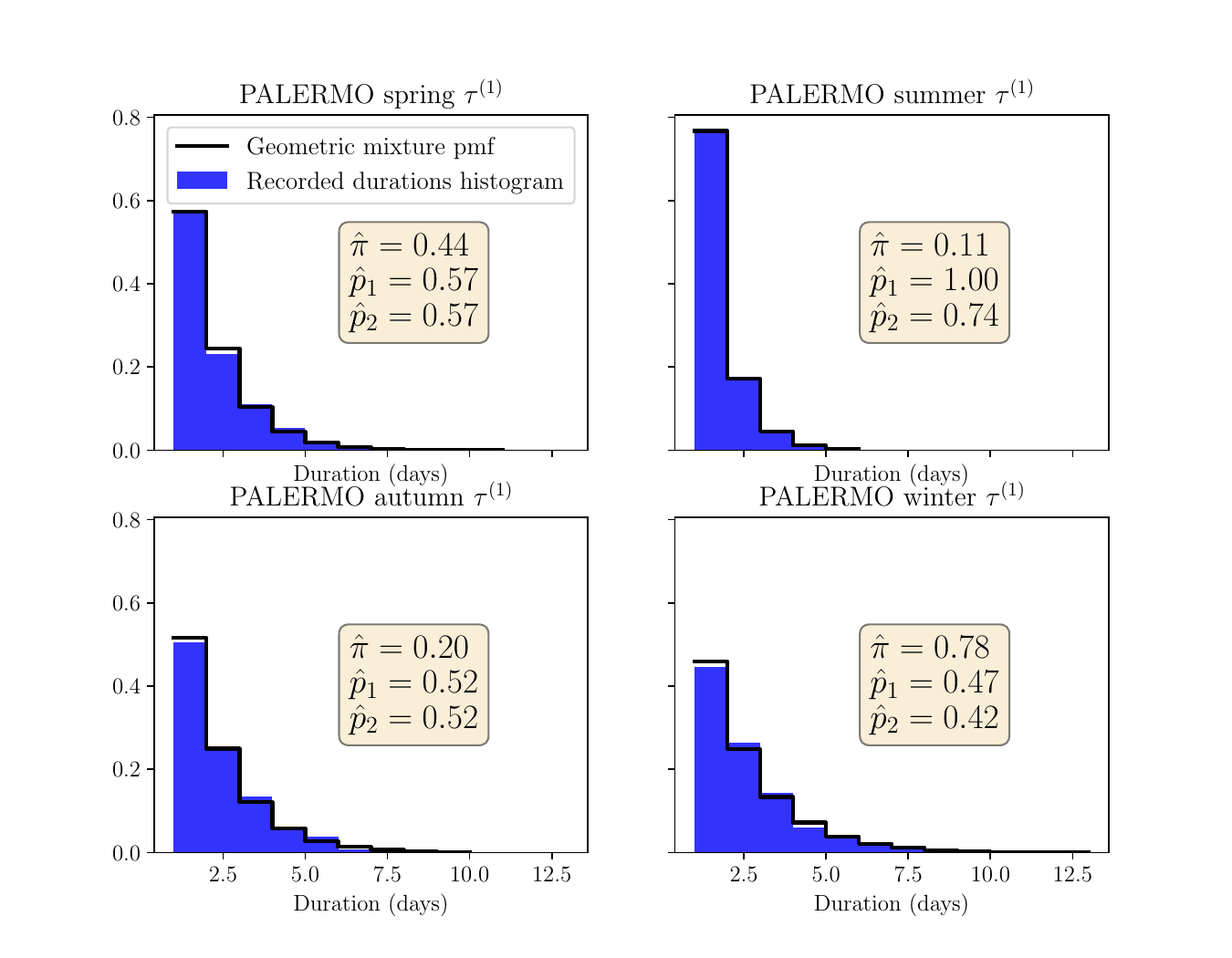}

    \caption{Histogram of recorded wet spell durations $\tau^{(1)}$ in Palermo and for every season. The black curve is the fitted probability mass function 
    \eqref{eq:mixgeom_pmf_wet} of the
    mixture of two geometric random variables. The fitted parameters are indicated in the legend.}
    \label{fig:histograms_fitted_geom_mix}
\end{figure}

\begin{figure}[!htbp]
    \centering
        \includegraphics[trim=25 20 25 25,clip,width=0.75\linewidth]{ 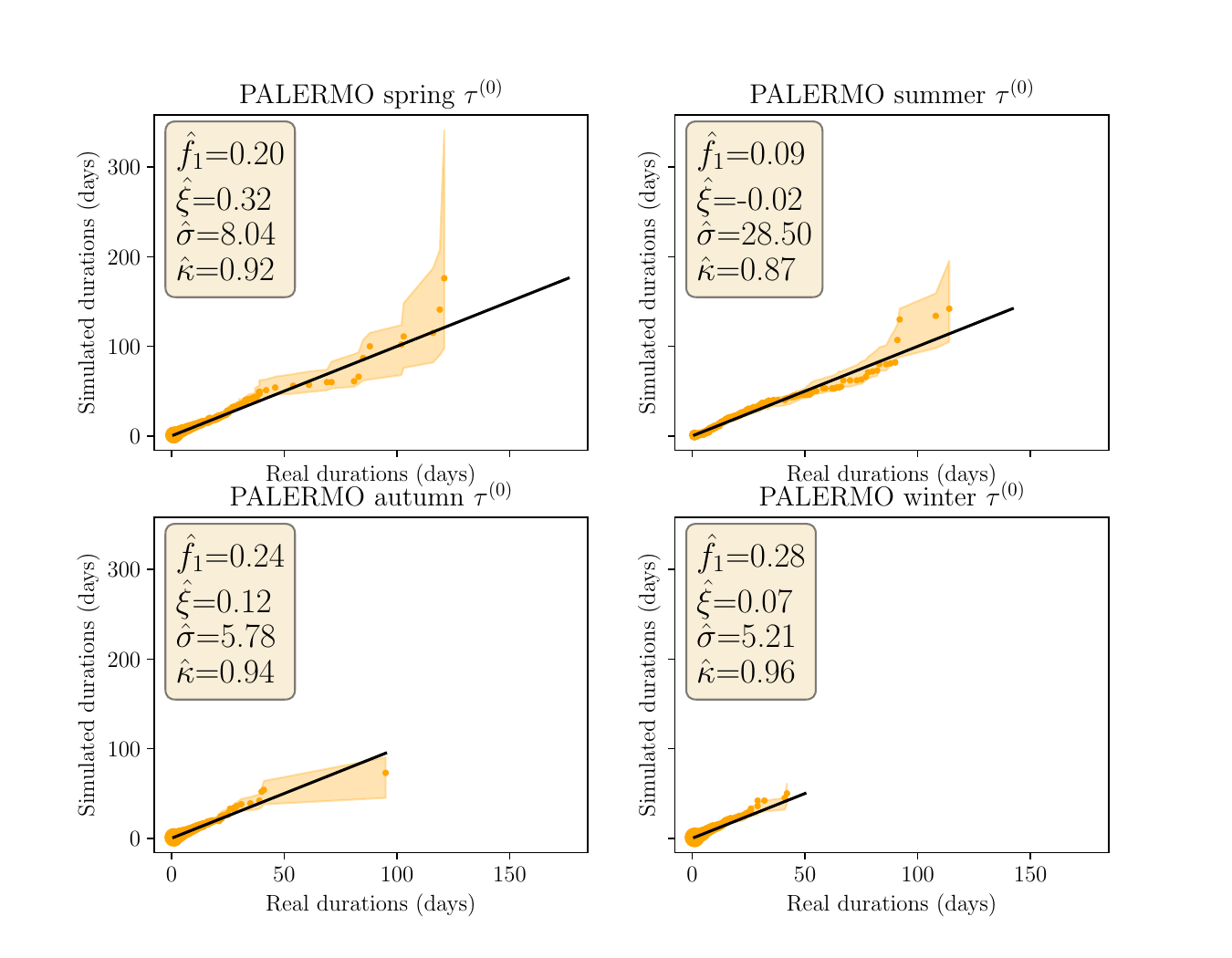}
    \caption{Q-Q plots (detailed in Appendix~\ref{subsec:simulation_based_qqplots_explanation}) for dry spell durations in Palermo and for every season. Sorted simulated dry spell durations $\tau^{(0)}$ from the fitted model \eqref{eq:tau0_pmf} are plotted against sorted recordings: the diagonal is equality. Bootstrap envelopes show sampling variability under the fitted model.}\label{fig:qqplot_fitted_dry}
\end{figure}

\begin{figure}[!htbp]
    \centering
        \includegraphics[trim=25 20 25 25,clip,width=0.75\linewidth]{ 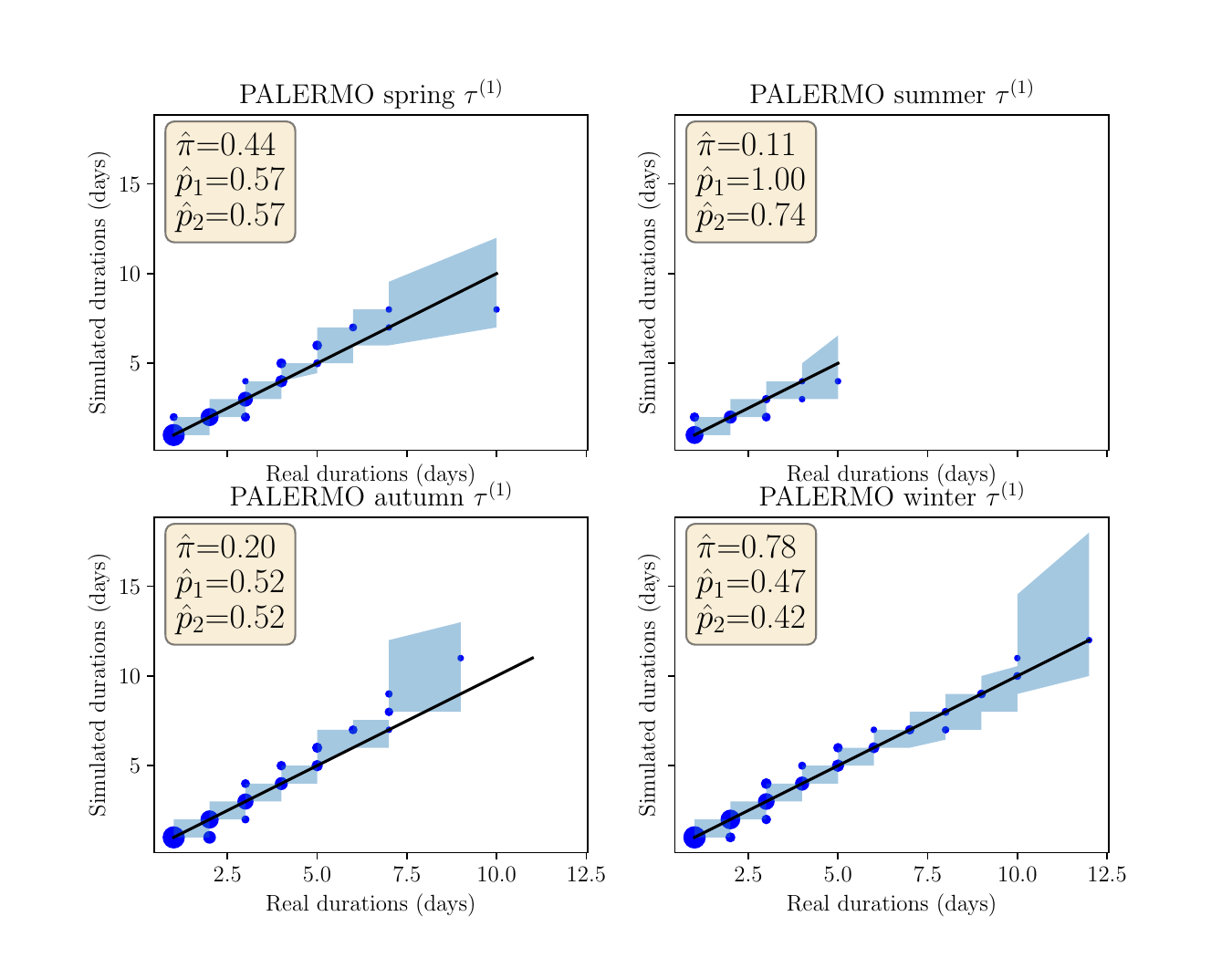}
    \caption{Q-Q plots (detailed in Appendix~\ref{subsec:simulation_based_qqplots_explanation}) for wet spell durations in Palermo and for every season. Sorted simulated wet spell durations $\tau^{(1)}$ from the fitted model \eqref{eq:mixgeom_pmf_wet} are plotted against sorted recordings: the diagonal is equality. Bootstrap envelopes show sampling variability under the fitted model. }\label{fig:qqplot_fitted_wet}    
\end{figure}

We now study how well our modelling reproduces the exit state probabilities, in particular the persistence of dry spells, that is, the fact that the exit probability from a dry spell tends to decrease with its current duration. To this end, we compare $\widehat{q}^{(0)}_{d,\hat\theta}$ and  $\widehat{q}^{(0)}_{d,\mathrm{emp}}$ introduced in equations \eqref{eq:link_tau_exit_proba} and \eqref{eq:exit_proba_empirical_estimator}, using the methods detailed in Section \ref{subsec:validation_tools}, in Palermo and for every season. Fig.~\ref{fig:proba_exit_state} displays estimates of the exit probabilities as functions of the elapsed duration (in days) since dry spell start, with $\widehat{q}^{(0)}_{d,\hat\theta}$ (black crosses) and $\widehat{q}^{(0)}_{d,\mathrm{emp}}$ (orange dots). Fig.~\ref{fig:proba_exit_state_belgrade} shows the same display for Belgrade (Serbia) where $\hat \xi < 0 $. The grey band indicates pointwise uncertainty under the fitted model. The empirical estimates reveal a characteristic two-phase behaviour. First, there is a quick decline for the first three to four days, which applies to every station. Then, there is a second part which is closely linked to the sign of the fitted $\hat \xi$. There can be a slow decrease (for Palermo in spring, autumn and winter, where $\hat \xi > 0 $), a nearly-constant slope (for Palermo in summer, where $\hat \xi \simeq 0 $) or an increase (for Belgrade in spring, where $\hat \xi < 0 $). The slow-decay regime is most cleanly illustrated by Palermo-spring (top-left panel of Fig.~\ref{fig:proba_exit_state}), in which $\widehat{q}^{(0)}_{d,\mathrm{emp}}$ falls from approximately $\widehat{q}^{(0)}_{1,\mathrm{emp}} = 0.2$ to $\widehat{q}^{(0)}_{5,\mathrm{emp}} = 0.1$ over the first $5$ days and then declines slowly to around $0.05$ after $30$-$40$ days. The near-constant regime is visible in Palermo-summer (top-right panel of Fig.~\ref{fig:proba_exit_state}): the initial quick decline is still present, but beyond $4$ days there is no clear downward trend. Finally, the rising regime of Belgrade-spring (Fig.~\ref{fig:proba_exit_state_belgrade}) exhibits the same quick initial decrease, but is thereafter characterised by an upward drift in $\widehat{q}^{(0)}_{d,\mathrm{emp}}$, indicating that longer dry spells become progressively easier to exit. Autumn and winter (bottom-left and bottom-right panels) also belong to the slow-decay regime, but display a slightly more irregular transition between the two phases. In autumn, for instance, the quick initial decrease persists over the first three days (from $\widehat{q}^{(0)}_{1,\mathrm{emp}} = 0.23$ to $\widehat{q}^{(0)}_{3,\mathrm{emp}} = 0.17$), but is followed by a transient rise ($\widehat{q}^{(0)}_{4,\mathrm{emp}} = 0.20$, $\widehat{q}^{(0)}_{5,\mathrm{emp}} = 0.19$) before the gradual decline resumes, reaching values around $\widehat{q}^{(0)}_{27,\mathrm{emp}} = 0.13$. A broadly similar pattern is observed in winter, and the phenomenon appears across almost all stations in both autumn and winter, suggesting a structural feature. We now assess how the estimates $\hat{q}_{d,\hat{\theta}}^{(0)}$, $d = 1, \ldots$ (black crosses) obtained under the hdeGPD model compare with the empirical values in each regime. The model reproduces the quick initial decrease throughout and accommodates all three long-duration behaviours: the slow asymptotic decay in Palermo-spring (top-left panel), the near-flat curve in Palermo-summer (top-right panel), and the gradual rise in Belgrade-spring The latter is consistent with the fitted $\hat\xi < 0$: a negative shape parameter imposes a finite upper bound on the spell duration, so the fitted exit probability rises with $d$ and reaches $1$ at that bound. It also provides a satisfactory fit for the two-phase decline observed in autumn and winter. However, the transient elevation in $\widehat{q}^{(0)}_{d,\mathrm{emp}}$ immediately following the initial quick decline is not captured by the hdeGPD distribution, which lacks the flexibility to accommodate this feature: this limitation could be investigated in future work. By fitting a spell-duration distribution with few parameters and recovering~\eqref{eq:link_tau_exit_proba}, rather than relying on a non-parametric estimator such as~\eqref{eq:exit_proba_empirical_estimator}, the approach yields smooth estimates and avoids the parameter instability noted by~\citet{1984_stern_model_fitting_daily_rainfall}. In summary, the empirical exit-probability sequences exhibit complex behaviour, with three qualitatively distinct long-duration regimes indexed by the sign of $\hat\xi$, yet the general trend is well reproduced by our model in each case. Recall that a classical finite-state Markov chain of order $d_0$ would yield a constant sequence $q_{d}^{(r)}$ for $d > d_0$, which would lack flexibility when $d_0$ is small, and would suffer from noisy estimates when $d_0$ is large. This underscores the advantage of the suggested approach.

In order to quantify the adequacy of the exit state probabilities modelling, we apply the chi-squared goodness-of-fit test of proposition~\ref{prop:gof_q} to dry spell durations, at each station and season. We consider two models: our hdeGPD-based specification and a model with dry spell duration following a geometric distribution (matching the common two-state first-order Markov model). The cut-off $d_{\max}$ is chosen adaptively so that at least $20$ dry spells are longer than $d_{\max}$ ($N_{n}^{(0)}(d_{\max}) \ge 20$), ensuring that the central limit theorem approximation used in proposition~\ref{prop:gof_q} remains valid. Fig.~\ref{fig:pvalues_map_south_europe_geom_vs_egpd} display the map of p-values for each season: from spring in the top panel to winter in the bottom panel. For each station, the p-values associated with the two models are displayed, namely the geometric model (large background circle), and the hdeGPD model (smaller foreground circle). Brown and red shades indicate smaller p-values (approximately 0 and 0.01), whereas yellow and green shades indicate larger p-values (approximately 0.5 and 1). Hence, for a given station and season, an improvement in exit-state modelling (non-rejection of the goodness-of-fit test in proposition~\ref{prop:gof_q} at the 1\% level) is reflected by a large brown circle containing a smaller circle in a lighter shade. The summer panel contains relatively few large brown circles, suggesting that the geometric specification already provides a reasonably adequate description of exit states during this season, which is consistent with the near-constant exit state highlighted in Fig.~\ref{fig:proba_exit_state}. In the other seasons, the north-eastern part of the domain, particularly around Romania and Bulgaria, also shows fewer large brown circles, again indicating a comparatively satisfactory fit of the geometric model in these regions. By contrast, large brown circles are common elsewhere, revealing poor performance of the geometric specification, especially in Spain, Portugal, and along the eastern Adriatic coast. Turning to the p-values for the hdeGPD model, one observes that a majority of the large brown circles contain smaller circles with lighter shades, whereas the reverse pattern is rare, meaning a reduced number of rejected stations. The stations for which the hdeGPD specification is still rejected by the goodness-of-fit test appear to be concentrated mainly in northern Portugal and Spain and, to a lesser extent, along the eastern Adriatic coast, where some improvement nevertheless remains visible. Although the hdeGPD distribution is flexible, it may still be unable to capture all station-specific spell-duration behaviours: allowing for a broader class of duration distributions could further improve the fit and reduce the remaining rejections.

\begin{figure}[!htbp]
    \centering
    \includegraphics[width=0.75\linewidth]{ 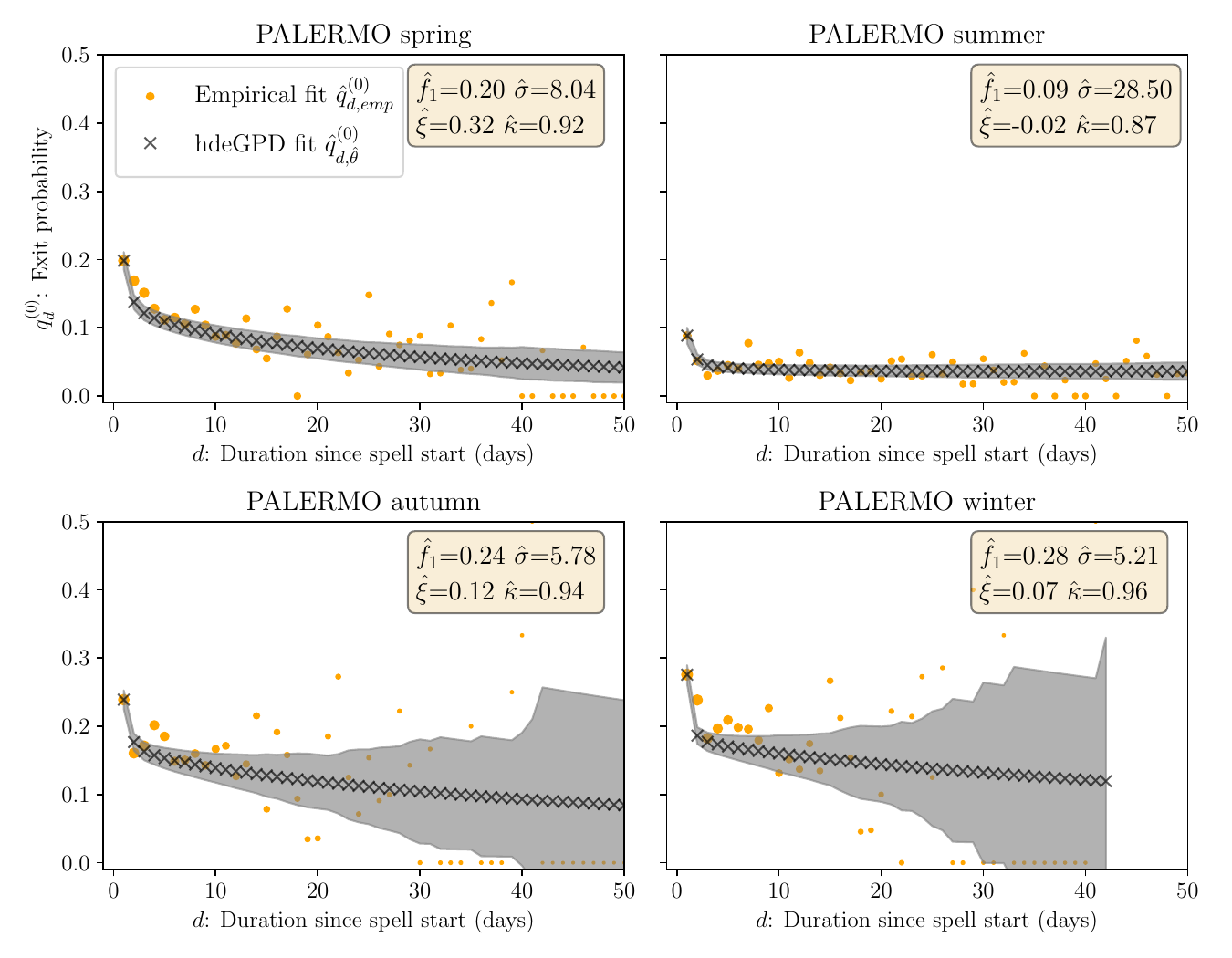}
    
    \caption{Exit probabilities for dry spells $\tau^{(0)}$ as a function of elapsed duration (in days) since spell start in Palermo and for every season. Empirical estimations (\textcolor{orange}{$\bullet$}), from equation \eqref{eq:exit_proba_empirical_estimator}, and model-implied estimations ($\times$), from equation \eqref{eq:link_tau_exit_proba}.  Shaded bands indicate pointwise uncertainty under the fitted model.}
\label{fig:proba_exit_state}
\end{figure}
\begin{figure}
    \centering
    \includegraphics[width=0.4\linewidth]{ 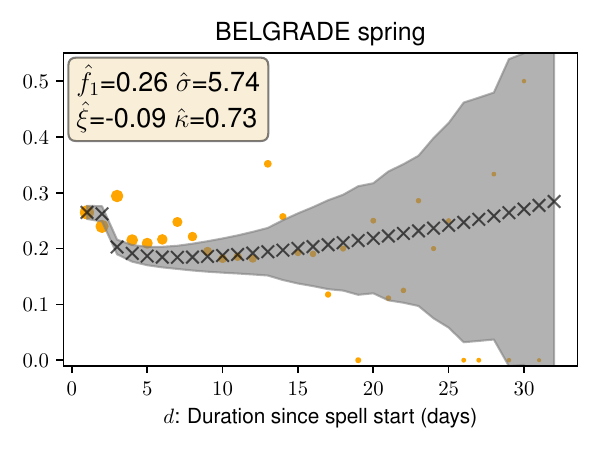}
    \caption{Same display as Fig.~\ref{fig:proba_exit_state}, for Belgrade (Serbia) in spring, for which $\hat \xi = -0.09$.}
    \label{fig:proba_exit_state_belgrade}
\end{figure}

\begin{figure}[!htbp]
\centering
\includegraphics[width=0.8\textwidth]{ 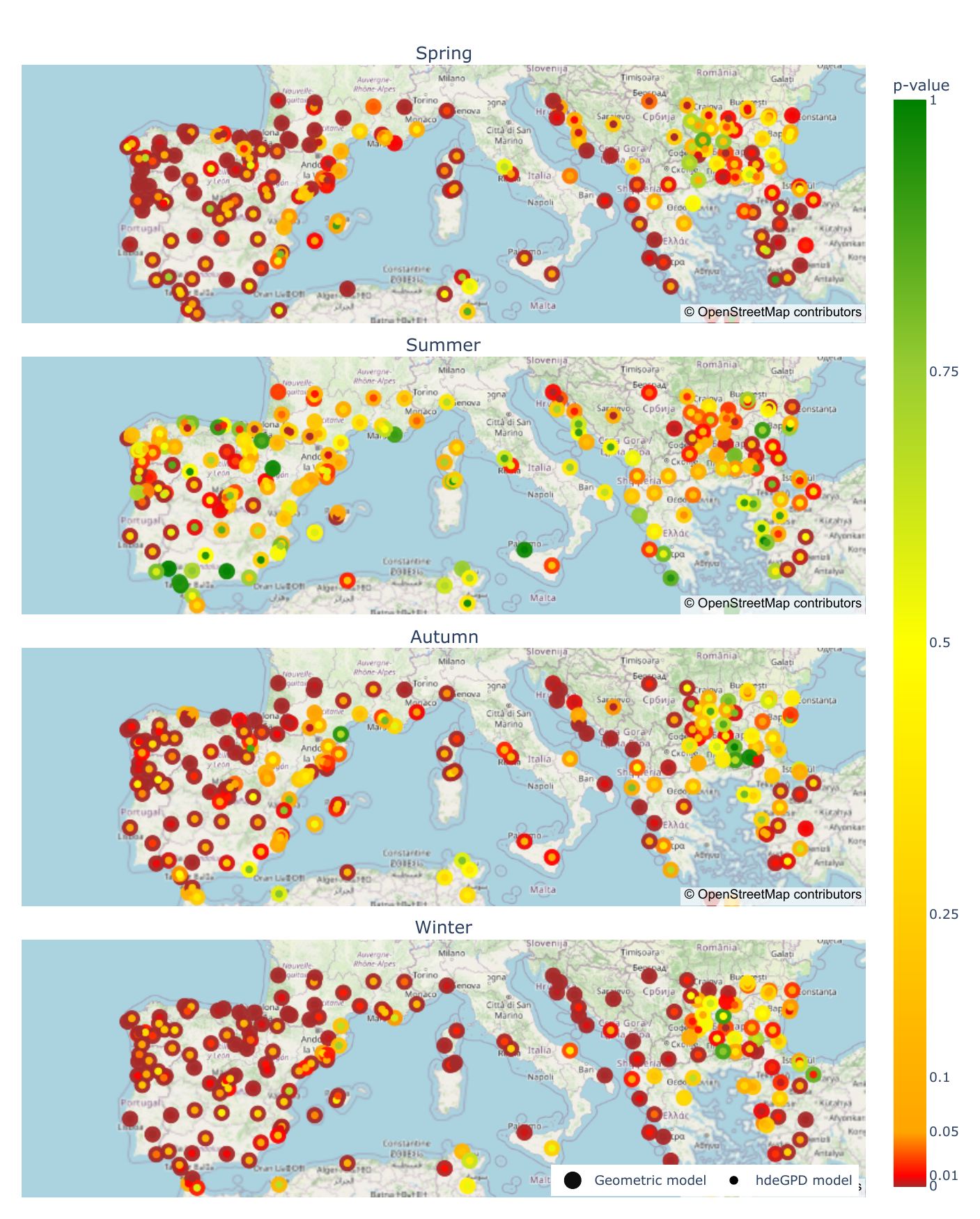}
\caption{Maps of p-values (brown-red near $0$, yellow-green between  $0.5$ and $1$) of the goodness-of-fit test (proposition~\ref{prop:gof_q}), under the geometric (large background circle) and the hdeGPD (smaller foreground circle) models in southern Europe for every season (from spring in top panel, to winter in bottom panel). A large brown circle containing a smaller not-brown circle means an improvement of the modelling (non-reject of the test at level $1\%$).}
\label{fig:pvalues_map_south_europe_geom_vs_egpd}
\end{figure}

\subsection{Improvement for extreme dry spell risk assessment}

The diagnostic results of Section~\ref{subsec:model_validation_and_robustness} confirm that, across most of the domain and seasons considered, the hdeGPD-based model from Section~\ref{subsec:spell_duration_distribution_specification} provides a substantially better description of spell duration behaviour than the geometric model arising from a two-state first-order Markov chain. In particular, it offers an improved fit in the tail, where the geometric distribution may notably underestimate the probability of extreme dry spells, as illustrated by the analysis of Figure~\ref{fig:survival_func_sample_data_intro}. We now examine how this improvement translates into the assessment of risk linked to these events.

Following standard practice in extreme value analysis \citep[Section 4.3.1]{coles2001introduction}, we consider the mean residual duration  of a dry spell after $d$ dry days which is the expected remaining duration of a dry spell which has already lasted for at least $d$ days:
\begin{equation}\label{eq:mean_residual_dry_spell_duration}
\mathbb{E}\bigl[\tau^{(0)}-d \mid \tau^{(0)}>d\bigr].
\end{equation}
The details of the estimation procedure are given in Appendix~\ref{subsec:expected_value_long_dry_spells}. In particular we use the bounds in equation~\eqref{eq:final_bounds_proportion_long_dry_spells_refined} to compute an approximation up to a precision of $10^{-5}$ for the hdeGPD model, and the closed-form expression from equation~\eqref{eq:mean_excess_duration_dry_spells_markov_case} for the geometric model. Fig.~\ref{fig:mean_residual_duration_long_dry_spells_south_europe} displays these quantities, for both models and for each station in southern Europe, during spring. Blue tones indicate near $0$ residual duration and yellow to red tones indicate residual duration between $50$ and $100$ days. From top to bottom, the three panels display maps computed with thresholds of 20, 40, and 60 consecutive dry days. Each station is represented by two concentric circles: the larger background circle corresponds to the geometric Markov model, while the smaller foreground circle corresponds to our model with hdeGPD specification. Two circles of notably different colours at a given station reflect a substantial difference in modelled exposure to prolonged dry spells. At the 20-day threshold, the two models produce broadly comparable estimates across most of the domains. However, a few stations, mostly in the South, exhibit warmer tones for the hdeGPD circle than for the geometric one, indicating a higher mean residual duration. This pattern is consistent with the tendency of the geometric model to underestimate the probability of long dry spells. As the threshold increases to 40 and then 60 days, a growing number of stations are affected by this discrepancy. As would be expected, the geometric model's underestimation of long dry spells becomes increasingly consequential at higher thresholds, where the tail behaviour of the distribution dominates. The largest discrepancies are concentrated along the Mediterranean coastline, where the hdeGPD specification yields mean residual durations substantially higher than those of the geometric model, which severely underestimate exposure to extreme dry conditions.
\begin{figure}[!htbp]
\centering
\includegraphics[width=0.8\textwidth]{ 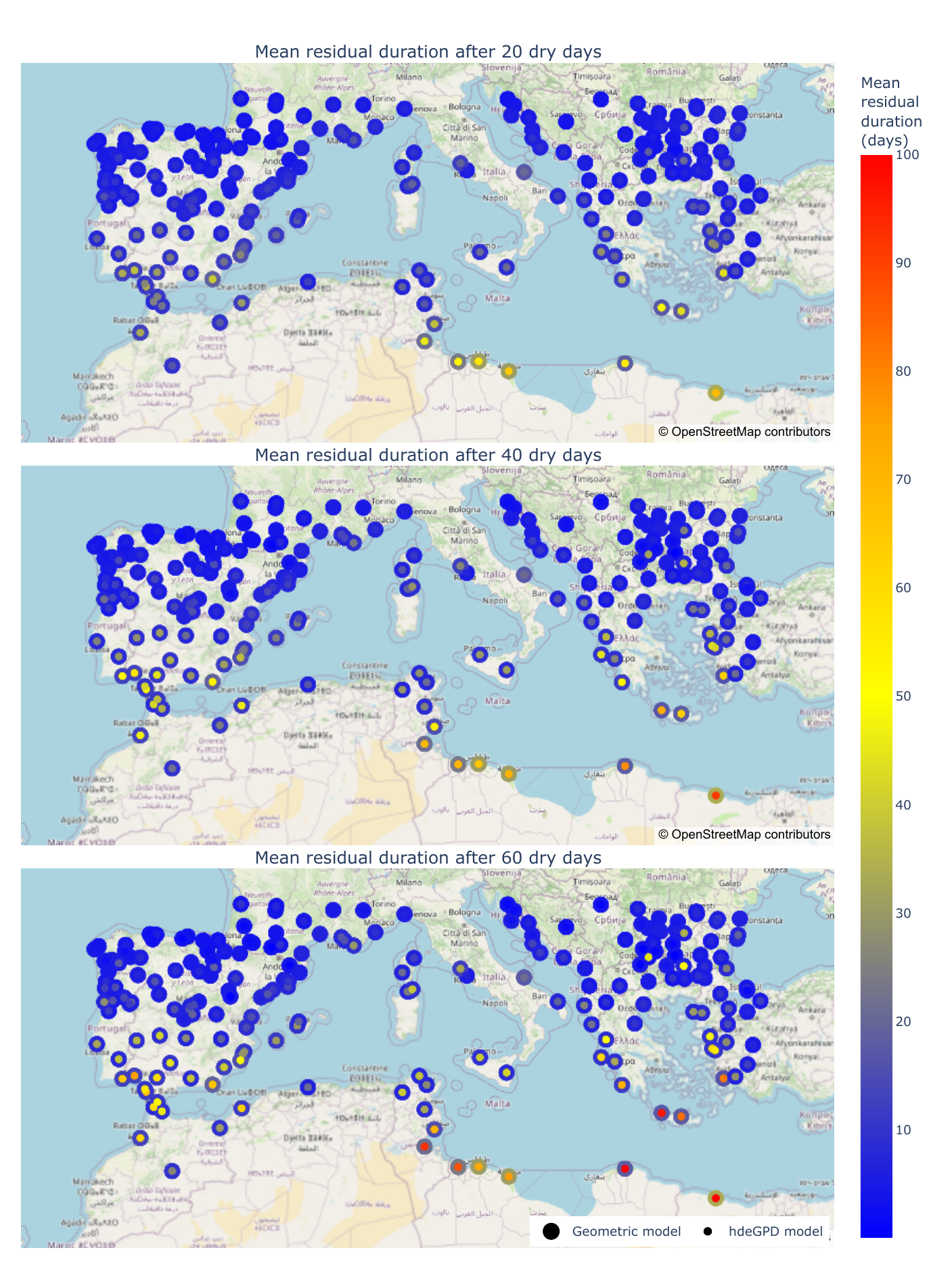}
\caption{Maps of mean residual duration of dry spell in spring, from equation \eqref{eq:mean_residual_dry_spell_duration}, under the geometric model (large background circle) and the hdeGPD model of Section \ref{subsec:spell_duration_distribution_specification} (smaller foreground circle), after (from top to bottom panel): 20, 40, and 60 dry days. Concentric circles with significant colour differences (blue for low to red for large durations) indicate substantially different predicted exposure to extreme dry spells.}
\label{fig:mean_residual_duration_long_dry_spells_south_europe}
\end{figure}

The diagnostics confirm the validity of the proposed framework. The fitted distributions agree well with the empirical data, the hdeGPD specification captures the declining exit probabilities of dry spells that the geometric model cannot reproduce, and this improvement translates into materially higher risk estimates for severe dry spells, especially along the Mediterranean coastline.

\section{Conclusion}\label{sec:conclusion}

This article introduced the Binary Markov Chain with Duration (BMCD) model, addressing the two blind spots identified in the introduction: a flexible joint modelling of the bulk and the tails of the spell duration distribution, and a formal equivalence between alternating renewal chains and a class of Markov chains. Applied to around 200 stations of the ECA\&D network in southern Europe, the model yielded three findings that illustrate the practical value of filling these gaps. First, no single tail behaviour describes dry spell durations across the region. Some stations display subexponential decay, while others exhibit markedly heavier tails or, conversely, behaviour consistent with a bounded support. Such heterogeneity falls outside the geometric family imposed by any finite-order Markov chain, and the eGPD specifications proved well suited to capturing it across stations and seasons. Second, we assessed the fit quality through histograms and Q-Q plots and then applying a goodness-of-fit test, applied both to the BMCD model and a two-state first-order Markov baseline. The latter showed an improvement of the fit mostly pronounced at Mediterranean stations and during spring, autumn, and winter. In summer, the simple baseline already produced satisfying results. Third, the improved tail modelling yielded revised predictions of exposure to long dry spells. The spring map reveals that the stations concentrated along the Mediterranean coast are the one for which the exposure was substantially underestimated by the Markov baseline.

Several limitations of the present work suggest natural extensions. The treatment of seasonality remains coarse: each spell is assigned to the season of its starting date, which introduces discontinuities at season boundaries and ignores any intra-seasonal evolution. A natural refinement would be to introduce a latent weather state on which the spell duration distribution is conditioned, thereby allowing a more flexible seasonal specification and smoothing the transition between regimes. The model also offers a convenient framework for the detection of climate change signals: covariates such as long-term temperature anomalies could be embedded directly into the eGPD parameters, or, alternatively, the model could be fitted on disjoint sub-periods and the resulting parameter estimates compared.
The most substantial direction for further work is the extension to multiple sites. The Markov representation, rather than the renewal one, makes this extension natural. The multisite framework of \citet{1998_wilks_multisite_stochastic_precipitation} is built around two-state first-order Markov chains, in which a constant exit probability governs spell duration (their equation (5)), and therefore inherits the underestimation of long dry spells documented above. Replacing this constant exit probability with a duration-dependent sequence, as in our construction, would yield multisite generators whose marginal spell duration distributions can match the heterogeneous tail behaviour observed in the data. Constructing an appropriate spatial dependence structure remains an open problem, which we leave to future work.

\section{Competing interests}
No competing interest is declared.

\section{Author contributions statement}

A.D. conceived the study, carried out the implementation and analysis, and wrote the first draft of the manuscript. D.A., P.N., and O.W. supervised the work, contributed to the interpretation of the results, and reviewed the manuscript.

\section{Data Availability Statement}\label{sec:data_availability_statement}
The code, data, and figures supporting this article are available in the following GitHub repository: \url{https://github.com/antoinedoize/rainfall_occurrence_BMCD}. The repository also contains code to reproduce the results presented in this article and to fit a BMCD to other rainfall datasets or other types of data.

The implementation uses map tiles from \texttt{plotly} based on data by 
\href{https://www.openstreetmap.org/copyright}{OpenStreetMap contributors}, 
available under the Open Database Licence (ODbL).

\section{Acknowledgments}
This work has been supported by the chair Geolearning, funded by ANDRA, BNP Paribas, CCR and the SCOR Foundation for Science. 
Part of Naveau’s research work was supported by the   Agence Nationale de la Recherche via:  the SICIM and SHARE PEPR Maths-Vives project (France 2030 ANR-24-EXMA-0008), EXSTA grant (ANR-23-CE40-0009-01), PORC-EPIC, the PEPR TRACCS program  (PC4 EXTENDING, ANR-22-EXTR-0005), and  the PEPR   IRIMONT (France 2030 ANR-22-EXIR-0003).

During the preparation of this work, the submitting author used the natural language processing model Claude Opus 4.7, accessed during April 2026, to correct typographical, grammatical, and syntactic errors, and to assist with the refactoring of the Python code used for the analyses. After using this tool, the authors reviewed and edited the content as needed and take full responsibility for the content of the publication.
\newpage

\bibliographystyle{abbrvnat}
\bibliography{references}





\begin{appendices}

\newpage

\section{Proofs}\label{sec:proofs}

\subsection{Preliminary results}\label{subsec:preliminary_results}

Here we state two classical results which can be found in any reference to alternating renewal models, such as \citet[theorems~2.3-2.4]{semi_markov_barbu}.

\begin{theorem}[Strong law of large numbers and elementary renewal theorem]\label{theorem:law_large_number_elemementary_renewal_theorem_N_n}
Suppose $\mathbb{E}[\tau] < \infty$. Then,
\begin{equation}\label{eq:law_large_number_N_n}
\lim_{n \to \infty} \frac{N(n)}{n} \;=\; \frac{1}{\mathbb{E}[\tau]} \quad \text{a.s.}
\end{equation}
We also have,
\[
\lim_{n \to \infty} \frac{\mathbb{E}\left[ N_n \right]}{n} \;=\; \frac{1}{\mathbb{E}[\tau]}.
\]
The first result is a strong law of large numbers for $(N_n)$, and the second result is commonly known as the elementary renewal theorem.
\end{theorem}

\begin{theorem}[Central limit theorem for $N(n)$]\label{theorem:central_limit_N_n}
Suppose $\mathrm{Var}(\tau) < \infty$. Then 
\[
\frac{N(n) - n/\mathbb{E}[\tau]}{\sqrt{n}\sqrt{\mathrm{Var}(\tau) / \mathbb{E}[\tau]^3}} 
\;\;\xrightarrow[n \to \infty]{\;\;d\;\;}\;\; \mathcal{N}(0,1).
\]
\end{theorem}

\subsection{Proof of proposition~\ref{prop:reward_convergence}}\label{proof-reward_convergence}
\begin{proof}
Note that $T_{N_n}$ is the start time of the $N_n$-th “dry-rain” cycle, in particular $T_{N_n}\le n<T_{N_n+1}$, so $T_{N_n}$ is the last renewal time before $n$. This random variable is convenient to link the Markov chain representation to the alternating renewal chain representation.

Let us first consider the almost sure convergence. The discrete version of \cite[proposition~3.4.1]{resnick1992adventures} yields
\[
\lim_{n\to\infty}\frac{1}{n}\sum_{k=0}^{T_{N_n}-1} w(R_k,D_k)
= \frac{\rho}{\mathbb{E}[\tau]}
\quad\text{a.s.}
\]
Write
\[
\frac{1}{n}\sum_{k=0}^{n} w(R_k,D_k)
= \frac{1}{n}\sum_{k=0}^{T_{N_n}-1} w(R_k,D_k)
\;+\; \frac{1}{n}\sum_{k=T_{N_n}}^{n}w(R_k,D_k).
\]
Let $A_n:=n-T_{N_n}$ be the backward recurrence time. Then
\[
\frac{1}{n}\sum_{k=T_{N_n}}^{n}w(R_k,D_k)
= \frac{1}{n}\sum_{k=0}^{A_n}w(R_{k + T_{N_n}},D_{k + T_{N_n}}).
\]
Moreover, almost surely,
\[
A_n < \tau_{N_n} := T_{N_{n+1}}-T_{N_n},
\]
and since $w\ge 0$,
\[
0
\leq \frac{1}{n}\sum_{k=0}^{A_n}w(R_{k + T_{N_n}},D_{k + T_{N_n}})
\leq \frac{1}{n}\sum_{k=0}^{\tau_{N_n}-1}w(R_{k + T_{N_n}},D_{k + T_{N_n}}),
\quad\text{a.s.}
\]
The right-hand side is finite a.s. by hypothesis, and
\[
\sum_{k=0}^{\tau_{N_n}-1}w(R_{k + T_{N_n}},D_{k + T_{N_n}})
\stackrel{d}{=}
\sum_{k=0}^{\tau-1}w(R_k,D_k),
\]
which has finite expectation by assumption. Hence
\[
\lim_{n\to\infty}\frac{1}{n}\sum_{k=0}^{\tau_{N_n}-1}w(R_{k + T_{N_n}},D_{k + T_{N_n}})
=0,\quad\text{a.s.}
\]
and therefore
\[
\lim_{n\to\infty}\frac{1}{n}\sum_{k=0}^{A_n}w(R_{k + T_{N_n}},D_{k + T_{N_n}})
=0,\quad\text{a.s.}
\]
Combining these limits gives
\[
\lim_{n\to\infty}\frac{1}{n}\sum_{k=0}^{n} w(R_k,D_k)
=
\lim_{n\to\infty}\frac{1}{n}\sum_{k=0}^{T_{N_n}-1} w(R_k,D_k)
=
\frac{\rho}{\mathbb{E}[\tau]},
\quad\text{a.s.}
\]

Now let us consider the asymptotic normality. Define the centered random variables
\[
Y_i := \sum_{k=T_{N_i}}^{T_{N_{i+1}}-1} w(R_k, D_k) - \rho,
\qquad i\ge 0,
\]
so that $(Y_i)_{i\ge 0}$ are i.i.d.\ with $\mathbb{E}[Y_i]=0$ and
$\mathrm{Var}(Y_i)=\nu^2<+\infty$ by assumption. Then
\[
\sum_{i=0}^{N_n} Y_i
=
\sum_{k=0}^{T_{N_n}-1} w(R_k, D_k) - N_n\rho.
\]
Recall that $\frac{N_n}{n}\to \frac{1}{\mathbb{E}[\tau]}$ almost surely, hence also in probability. By Anscombe's theorem \citep[theorem 7.3.2]{chung2000course},
\[
\frac{\sum_{i=0}^{N_n}Y_i}{\sqrt{n\,\nu/\mathbb{E}[\tau]}}
=
\frac{\sum_{k=0}^{T_{N_n}-1} w(R_k, D_k) - N_n\rho}{\sqrt{n\,\nu/\mathbb{E}[\tau]}}
\;\xrightarrow[n\to\infty]{\ d\ }\;
\mathcal{N}(0,1).
\]
It remains to replace $\sum_{k=0}^{T_{N_n}-1} w(R_k,D_k)$ by $\sum_{k=0}^{n} w(R_k,D_k)$.
Using again $A_n=n-T_{N_n}$ and the positivity of $w$,
\[
0
\leq \sum_{k=T_{N_n}}^{n} w(R_k,D_k)
=
\sum_{k=0}^{A_n} w(R_{k+T_{N_n}},D_{k+T_{N_n}})
\leq
\sum_{k=0}^{\tau_{N_n}-1} w(R_{k+T_{N_n}},D_{k+T_{N_n}}),
\quad\text{a.s.}
\]
As above,
\[
\sum_{k=0}^{\tau_{N_n}-1} w(R_{k+T_{N_n}},D_{k+T_{N_n}})
\stackrel{d}{=}
\sum_{k=0}^{\tau-1} w(R_k,D_k),
\]
so in particular it is finite a.s. and has finite expectation. Therefore,
\[
\lim_{n\to\infty}\frac{1}{\sqrt{n}}
\sum_{k=T_{N_n}}^{n} w(R_k,D_k)=0,
\quad\text{a.s.}
\]
and Slutsky theorem \citep[theorem 11.4]{gut2006probability} yields
\[
\frac{\sum_{k=0}^{n} w(R_k, D_k) -N_n\rho}{\sqrt{n\,\nu^2/\mathbb{E}[\tau]}}
\;\xrightarrow[n\to\infty]{\ d\ }\;
\mathcal{N}(0,1),
\]
which is the stated asymptotic normality.
\end{proof}

The result could be extended to more general functions $w:(\{0,1\}\times\mathbb{N}^*)^\infty \to \mathbb{R}_+$, with $(\{0,1\}\times\mathbb{N}^*)^\infty := \{ (r_k,d_k)_{k \leq n}, \; (r_k,d_k) \in \{0,1\}\times\mathbb{N}^* , \; n \in \mathbb{N}^* \}$ the set of finite sequences takings values in $E$, but we keep this one for the sake of simplicity.

\subsection{Proof of proposition~\ref{prop:cltintermediate}}\label{proof:prop_cltintermediate}

\begin{proof}
The result follows from an application of Proposition
\ref{prop:reward_convergence}.

Let us set an integer $d_{\max}>2$ and define, for
$\boldsymbol{\Lambda}=(\lambda_d)_{2\le d\le d_{\max}}\in\mathbb{R}^{d_{\max}-1}$,
the function $w^{(r)}_{\boldsymbol{\Lambda}}(R_k,D_k)
=
\sum_{d=2}^{d_{\max}}\lambda_d\,\mathbbm{1}(R_k=r,\;D_k=d)
.$
Then
\[
\sum_{k=0}^{\tau-1} w^{(r)}_{\boldsymbol{\Lambda}}(R_k,D_k)
=
\sum_{d=2}^{d_{\max}}\lambda_d
\sum_{k=0}^{\tau-1}\mathbbm{1}(R_k=r,\;D_k=d)
=
\sum_{d=2}^{d_{\max}}\lambda_d\,\mathbbm{1}(\tau^{(r)}>d-1),
\]
as $\sum_{k=0}^{\tau-1}\mathbbm{1}(R_k=r,\;D_k=d)$ equals $1$ if and only if the r-type spell is longer or equal to $d$. Therefore
\[
\mathbb{E}\!\left[\sum_{k=0}^{\tau-1} w^{(r)}_{\boldsymbol{\Lambda}}(R_k,D_k)\right]
= \sum_{k=2}^{d_{\max}} \lambda_d \mathbb{E}\!\left[\,\mathbbm{1}(\tau^{(r)}>d-1)\right]
=
\sum_{d=2}^{d_{\max}}\lambda_d\,\overline F_{\tau^{(r)}}(d-1).
\]
Applying proposition~\ref{prop:reward_convergence} with $w = w^{(r)}_{\boldsymbol{\Lambda}}$ yields

\[
\frac{\sum_{k=0}^{n} \left[ w^{(r)}_{\boldsymbol{\Lambda}}(R_k, D_k)\right] - N_n \left[  \sum_{d=2}^{d_{\max}}\lambda_d\,\overline F_{\tau^{(r)}}(d-1) \right] }{\sqrt{n\,\nu^2(\boldsymbol{\Lambda})/\mathbb{E}[\tau]}} \;\xrightarrow[n\to\infty]{d}\; \mathcal{N}(0,1).
\]
Noting that $\sum_{k=0}^{n} \mathbbm{1}(R_k=r,\,D_k=d) = N_{n}^{(r)}(d)$,

\begin{equation}\label{eq:pre_cramer_wold_optionA}
\frac{
\sum_{d=2}^{d_{\max}}\lambda_d
\Big(N_{n}^{(r)}(d) - N_n\,\overline F_{\tau^{(r)}}(d-1)\Big)
}{
\sqrt{n\,\nu^2(\boldsymbol{\Lambda})/\mathbb{E}[\tau]}
}
\;\xrightarrow[n\to\infty]{d}\;
\mathcal{N}(0,1),
\end{equation}
where $
\nu^2(\boldsymbol{\Lambda})
=
\mathrm{Var}\!\left[
\sum_{k=0}^{\tau-1} w^{(r)}_{\boldsymbol{\Lambda}}(R_k,D_k)
\right]$. Let us denote just for this variance calculation the vector $\mathbf{X}=(X_d)_{2\le d\le d_{\max}}^{\mathsf T}$, defined by
$X_d=\sum_{k=0}^{\tau-1}\mathbbm{1}(R_k=r,\;D_k=d),
\qquad d=2,\dots,d_{\max}$.
Then
\[
\nu^2(\boldsymbol{\Lambda})
=
\mathrm{Var}\!\left(\boldsymbol{\Lambda}^{\mathsf T}\mathbf{X}\right)
=
\boldsymbol{\Lambda}^{\mathsf T}\,\Sigma\,\boldsymbol{\Lambda},
\]
where $\Sigma=\big(\Sigma_{i,j}\big)_{1\le i,j\le d_{\max}-1}$ is the covariance matrix of
$\mathbf{X}$.
For $1\le i,j\le d_{\max}-1$,
\begin{align*}
\Sigma_{i,j}
&=
\mathrm{Cov}\!\left(
\sum_{k=0}^{\tau-1}\mathbbm{1}(R_k=r,\;D_k=i+1),\,
\sum_{k=0}^{\tau-1}\mathbbm{1}(R_k=r,\;D_k=j+1)
\right) \\
&=
\mathrm{Cov}\!\left(
\mathbbm{1}(\tau^{(r)}>i),\,
\mathbbm{1}(\tau^{(r)}>j)
\right) \\
&=
\overline F_{\tau^{(r)}}\!\big(\max(i,j)\big)
-
\overline F_{\tau^{(r)}}(i)\,
\overline F_{\tau^{(r)}}(j).
\end{align*}
Rewriting \eqref{eq:pre_cramer_wold_optionA}, we obtain for any
$\boldsymbol{\Lambda}\in\mathbb{R}^{d_{\max}-1}$,
\[
\sum_{d=2}^{d_{\max}}\lambda_d
\sqrt{\frac{n}{\mathbb{E}[\tau]}}
\left(
\frac{N_{n}^{(r)}(d)}{N_n}
-
\overline F_{\tau^{(r)}}(d-1)
\right)
\;\xrightarrow[n\to\infty]{d}\;
\mathcal{N}\!\left(0,\boldsymbol{\Lambda}^{\mathsf T}\Sigma\boldsymbol{\Lambda}\right).
\]
Since this holds for every $\boldsymbol{\Lambda}\in\mathbb{R}^{d_{\max}-1}$,
the Cram\'er-Wold theorem yields
\[
\sqrt{\frac{n}{\mathbb{E}[\tau]}}
\Bigg(
\frac{N_{n}^{(r)}(d)}{N_n}
-
\overline F_{\tau^{(r)}}(d-1)
\Bigg)_{2\le d\le d_{\max}}
\;\xrightarrow[n\to\infty]{d}\;
\mathcal{N}_{d_{\max}-1}(\mathbf{0},\Sigma).
\]
Finally, plugging in
\eqref{theorem:law_large_number_elemementary_renewal_theorem_N_n}
and applying Slutsky's theorem \citep[theorem 11.4]{gut2006probability} concludes the proof.
\end{proof}

\subsection{Proof of proposition~\ref{prop:gof_q}}\label{proof:prop_gof_q}

\begin{proof} \label{proof-gof}

By proposition~\ref{prop:cltintermediate}, under $H_0^{(r)}$ the vector
\[
\left(\frac{N_{n}^{(r)}(d)}{N_n}-\overline F_{\theta_0}(d-1)\right)_{2\le d\le d_{\max}}
\]
satisfies a multivariate CLT with asymptotic covariance matrix $\Sigma_{\theta_0}$.
For $d=1,\dots,d_{\max}-1$, we have:
\[
q^{(r)}_1 \;=\; 1-\overline F_{\theta_0}(1),
\qquad
q^{(r)}_d \;=\; 1-\frac{\overline F_{\theta_0}(d)}{\overline F_{\theta_0}(d-1)},\quad d\ge 2,
\]
and similarly the empirical estimators satisfy
\[
\widehat q^{(r)}_{1,\mathrm{emp}}
\;=\;
1-\frac{N_{n}^{(r)}(2)}{N_n},
\qquad
\widehat q^{(r)}_{d,\mathrm{emp}}
\;=\;
1-\frac{N_{n}^{(r)}(d+1)/N_n}{N_{n}^{(r)}(d)/N_n},\quad d\ge 2,
\]
since $N_{\ge 1}^{(r)}=N_n$.
Define the smooth mapping $\varphi:\mathbb{R}^{d_{\max}-1}\to\mathbb{R}^{d_{\max}-1}$ by
\[
\varphi(x_2,\dots,x_{d_{\max}})
=
\Bigl(1-x_2,\ \ (1-x_{d+1}/x_d)_{d=2,\dots,d_{\max}-1}\Bigr).
\]
Applying the multivariate $\Delta$-method to the CLT from proposition~\ref{prop:cltintermediate} yields
\[
\sqrt{\frac{n}{\mathbb{E}[\tau]}}
\left(\widehat{\mathbf{q}}^{(r)}_{\mathrm{emp}}-\mathbf{q}_{\theta_0}^{(r)}\right)_{|1:d_{\max}-1}
\;\xRightarrow{d}\;
\mathcal{N}_{d_{\max}-1}\!\left(\mathbf{0},\;\mathbf{T}_{\theta_0}\Sigma_{\theta_0}\mathbf{T}_{\theta_0}^{\mathsf T}\right),
\]
where the Jacobian matrix is exactly $\mathbf{T}_{\theta_0}$ as defined in the proposition.
If $\mathbf{T}_{\theta_0}\Sigma_{\theta_0}\mathbf{T}_{\theta_0}^{\mathsf T}$ is nonsingular, the standard Gaussian
quadratic-form result implies that
\[
\frac{n}{\mathbb{E}[\tau]}\,
\mathbf{\Delta}^{\mathsf T}
\left(\mathbf{T}_{\theta_0}\Sigma_{\theta_0}\mathbf{T}_{\theta_0}^{\mathsf T}\right)^{-1}
\mathbf{\Delta}
\;\xRightarrow{d}\;
\chi^2_{ d_{\max}-1}.
\]
with $\mathbf{\Delta} := \bigl(\widehat {\mathbf{q}}_{\mathrm{emp}|1:d_{\max}-1}^{(r)}-\widehat {\mathbf{q}}_{\theta_0 |1:d_{\max}-1}^{(r)}\bigr)$
Using Slutsky theorem \citep[theorem 11.4]{gut2006probability}, along with theorem~\ref{theorem:central_limit_N_n}, we get
\[
\mathcal{Q}_{N_n} :=N_n\,
\mathbf{\Delta}^{\mathsf T}
\left(\mathbf{T}_{\theta_0}\Sigma_{\theta_0}\mathbf{T}^{\mathsf T}_{\theta_0}\right)^{-1}
\mathbf{\Delta}
\;\stackrel{d}{\longrightarrow}\;
\chi^2_{ d_{\max}-1}.
\]

\end{proof}

\section{Additional results}\label{sec:supplementary_material}

\subsection{Finite spell duration condition}\label{subsec:spell_durations_finiteness}

For a BMCD $(R_n,D_n)$ with parameters $(\mathbf q^{(0)},\mathbf q^{(1)})$, we will always consider first dry and wet spell durations which are a.s.\ finite. This happens if and only if the following condition on the parameters is fulfilled:
$$\sum_{d=1}^{\infty} q^{(0)}_d = +\infty \;\;\text{and}\;\; \sum_{d=1}^{\infty} q^{(1)}_d = +\infty.$$
Indeed, as \[
\mathbb P(\tau^{(0)}_1 \ge K)
\;=\;
\prod_{d=1}^{K} (1-q^{(0)}_d), \text{ by independence.}
\]
The product $  ( \prod_{d=1}^{K} (1-q^{(0)}_d)  )_{K \geq 1} $ has factors in $(0,1)$, so it is decreasing and lower bounded: it is converging either towards 0, or towards $l \in (0,1)$. We have:
\[
\mathbb P(\tau^{(0)}_1=\infty)
\;=\;
\lim_{K\to\infty}\mathbb P(\tau^{(0)}_1\ge K)
\;=\;
\prod_{d=1}^{\infty} (1-q^{(0)}_d).
\]
Moreover,
\[
\prod_{d=1}^{\infty} (1-q^{(0)}_d) = l \Leftrightarrow \exp( \sum_{d=1}^{\infty}\log(1-q^{(0)}_d)) = l \Leftrightarrow \sum_{d=1}^{\infty}\log(1-q^{(0)}_d) \text{ converges,}
\]
and thus $(q^{(0)}_d)_{d \geq 1}$ being necessarily converging toward $0$. Then, applying the limit comparison test \[
\prod_{d=1}^{\infty} (1-q^{(0)}_d) = l \;\Leftrightarrow\; \sum_{d=1}^{\infty}q^{(0)}_d \text{ converges.}
\]
Therefore $\tau^{(0)}_1<\infty$ a.s. if and only if $\sum_{d=1}^{\infty} q^{(0)}_d \;=\; \infty$. With the same proof applied to the time shifted  Markov chain $(R_{n+\tau^{(0)}_1},D_{n+\tau^{(0)}_n})_{n \geq 0}$ we have that $\tau^{(1)}_1<\infty$ a.s. if and only if $\sum_{d=1}^{\infty} q^{(1)}_d \;=\; +\infty$.

\subsection{Spell duration i.i.d. property}\label{subsec:spell_durations_iid}

For all $r\in\{0,1\}$ and all $k\ge1$, the random variables $\{\tau^{(r)}_k\}$ are mutually independent. Moreover, $(\tau^{(0)}_k)_{k\ge1}$ are i.i.d.\ with common distribution denoted $\tau^{(0)}\stackrel d=\tau^{(0)}_1$, and $(\tau^{(1)}_k)_{k\ge1}$ are i.i.d.\ with common distribution denoted $\tau^{(1)}\stackrel d=\tau^{(1)}_1$. Consequently, the cycle durations $(\tau_k)_{k\ge1}$ are i.i.d.\ with common distribution denoted $\tau\stackrel d=\tau_1$.

This is a direct consequence of the strong Markov property at stopping times \[
\tau^{(0)}_1,\; \tau^{(0)}_1+\tau^{(1)}_1,\; \tau^{(0)}_1+\tau^{(1)}_1+\tau^{(0)}_2,\; \tau^{(0)}_1+\tau^{(1)}_1+\tau^{(0)}_2+\tau^{(1)}_2,\;\dots.
\]

\subsection{EM algorithm}\label{subsec:em_algorithm_mixture_geometric}
The EM algorithm is a simple and efficient way to estimate the parameters of mixture models such as the distribution described in equation~\eqref{eq:mixgeom_pmf_wet} \citep{mengersen2011mixtures}. Let $d_1,\dots,d_n$ denote an independent sample from the distribution \eqref{eq:mixgeom_pmf_wet}. Introduce latent indicators $z_k\in\{0,1\}$ such that, conditionally on $z_k=1$, $d_k$ follows a geometric distribution with parameter $p_1$, whereas, conditionally on $z_k=0$, $d_k$ follows a geometric distribution with parameter $p_2$. Then
\[
\mathbb{P}(z_k=1)=\pi,\qquad \mathbb{P}(z_k=0)=1-\pi,
\]
and the complete-data log-likelihood is, up to an additive constant,
\[
\begin{array}{rcl}
\ell_c(\pi,p_1,p_2)
=  & \sum_{k=1}^n z_k & \Big[\log\pi+\log p_1+(d_k-1)\log(1-p_1)\Big] \\
+ & \sum_{k=1}^n (1-z_k) &\Big[\log(1-\pi)+\log p_2+(d_k-1)\log(1-p_2)\Big].
\end{array}
\]
At iteration $m$, given current values $\pi^{(m)},p_1^{(m)},p_2^{(m)}$, the E-step replaces $z_k$ by its conditional expectation
\[
w_k^{(m)}
=\mathbb{E}(z_k\mid d_k)
=\frac{\pi^{(m)} p_1^{(m)}(1-p_1^{(m)})^{d_k-1}}
{\pi^{(m)} p_1^{(m)}(1-p_1^{(m)})^{d_k-1}
+(1-\pi^{(m)}) p_2^{(m)}(1-p_2^{(m)})^{d_k-1}}.
\]
The M-step maximizes the conditional expectation of $\ell_c$, which yields the closed-form updates
\[
\pi^{(m+1)}=\frac{1}{n}\sum_{k=1}^n w_k^{(m)},
\]
\[
p_1^{(m+1)}
=\frac{\sum_{k=1}^n w_k^{(m)}}
{\sum_{k=1}^n w_k^{(m)} d_k},
\qquad
p_2^{(m+1)}
=\frac{\sum_{k=1}^n \big(1-w_k^{(m)}\big)}
{\sum_{k=1}^n \big(1-w_k^{(m)}\big) d_k}.
\]
The E-step and M-step are iterated until convergence (relative increment of log-likelihood lower than $10^{-6}$). As usual for finite mixtures, random starting values are considered in order to reduce the risk of convergence to a local maximum.

\subsection{Autocorrelation of spell durations}\label{subsec:autocorrelation_bivariate_spell_duration}

Below is detailed a way of checking independence of consecutive spell durations of a BMCD explained in ~\ref{subsec:spell_durations_iid}. Rather than attempting a fully nonparametric test of independence, we suggest using the sample autocorrelation function as a simple, interpretable diagnostic for stationary time series. See \citet[Section 4.4]{Lutkepohl2005NewIntroduction} for an introduction to correlation matrix. Because our data consist of several disjoint series (one per year and season), we adapt this correlation matrix calculation by pooling valid within-series pairs at each lag. We detail how to define this matrix for a given season (not denoted here for readability). Let us reorder in this Section the vector of all spell durations of a given season by blocs of sequences for each year $y$. For each year $y = 1 \dots Y$ denote those reordered spell duration couples $\mathbf{V}_{y,k}=(\tau^{(0)}_{y,k},\tau^{(1)}_{y,k})^\top$, for $k=1,\dots,C_y$ with $C_y$ the number of couples for year $y$. Define
\[
C_{\mathrm{tot}}=\sum_{y=1}^Y C_y,\quad
\overline{\mathbf{V}}=\frac1{C_{\mathrm{tot}}}\sum_{y=1}^Y\sum_{k=1}^{C_y}\mathbf{V}_{y,k},\quad
C_\ell=\sum_{y=1}^Y\max(C_y-\ell,0),
\]
\[
\widehat{\mathbf{\Gamma}}(\ell)=\frac1{C_\ell}\sum_{y=1}^Y\sum_{k=1}^{C_y-\ell}(\mathbf{V}_{y,k}-\overline{\mathbf{V}})(\mathbf{V}_{y,k+\ell}-\overline{\mathbf{V}})^\top,
\quad \ell=0,\dots,L.\]
Then,
\begin{equation}\label{autocorrelation_matrix_several_time_series}
\widehat{\mathbf{R}}(\ell)=\widehat{\mathbf{D}}^{-1/2}\,\widehat{\mathbf{\Gamma}}(\ell)\,\widehat{\mathbf{D}}^{-1/2},
\end{equation}
where $\widehat{\mathbf{D}}=\mathrm{diag}(\widehat{\mathbf{\Gamma}}(0))$, is the autocorrelation matrix for lag $ \ell=0,\dots,L$. A very classical idea is to plot each of the correlation matrix elements for a given lag with the confidence bands $\ell \to \frac{2}{\sqrt{C_\ell}}$. This matrix is not symmetric in general.

\subsection{General class of extended Generalized Pareto Distribution (eGPD)}\label{subsec:general_class_egpd}

We now show that the two spell-duration models used in the paper are both obtained as special cases of the same general eGPD construction. For $z\in\mathbb{R}$, define
\begin{equation}\label{eq:generalized_pareto_cdf}
H_\xi(z)=
\begin{cases}
1-(1+\xi z)^{-1/\xi}, & (\xi>0, z\ge0)\ \text{or}\ (\xi<0,0<z<-1/\xi),\\[2pt]
1-e^{-z}, & \xi=0,\ \ z\ge 0,\\[2pt]
0, & \text{otherwise.}
\end{cases}    
\end{equation}

A random variable $X$ is said to belong to the general class of eGPD if we can find a cumulative distribution function $B(\cdot)$ on the unit interval with bounded density such that: 
$$X = \sigma H^{-1}_{\xi} \left[\left( B^{-1} (U) \right)^{1/\kappa}\right].$$ with $U$ a standard uniform distributed random variable, $\sigma > 0, \; \kappa > 0$, and $H^{-1}_{\xi}$ denoting the generalized inverse distribution function of $H_{\xi}$, such that for any $p \in [0,1], \; H^{-1}_{\xi} = \inf\{x \in \mathbb{R}: H_{\xi} \ge p \}$.
Thus, for any $x \in \mathbb{R}$ we have $$F_X(x) = \mathbb{P} \left( X \le x \right) = B \left( H^{\kappa}_{\xi} (\frac{x}{\sigma}) \right)$$.

In order to link this general expression to dry spell duration distribution from Section~\ref{subsec:spell_duration_distribution_specification}, consider
\[
B(u)=u \quad\text{for }u\in[0,1]
\]
yields $B^{-1}(u)=u$ and therefore
\[
F_X(x)=B\!\left(H_\xi^\kappa\!\left(\frac{x}{\sigma}\right)\right)
=H_\xi^\kappa\!\left(\frac{x}{\sigma}\right)
=F_{\kappa,\sigma,\xi}(x),
\]
which is the distribution of \eqref{eq:type-1-eGPD}. Then, by using
\[
B(u)=
\begin{cases}
f_1, & u=0,\\
f_1+(1-f_1)u, & 0<u\le 1,
\end{cases}
\]
we get:
\[
F_X(x)=B\!\left(H_\xi^\kappa\!\left(\frac{x}{\sigma}\right)\right)
=
\begin{cases}
f_1, & x=0,\\
f_1+(1-f_1)\,F_{\kappa,\sigma,\xi}(x), & x>0,
\end{cases}
\]
and using $\tau^{(0)} \coloneqq 1 + \lceil X \rceil$, we have,
\begin{equation*}
\mathbb{P}\big(\tau^{(0)}=d\big)=
\begin{cases}
f_1, & d=1,\\
(1-f_1)\,\Big[F_{\kappa,\sigma,\xi}(d-1)-F_{\kappa,\sigma,\xi}(d-2)\Big], & d\ge 2.
\end{cases}
\end{equation*}

We now turn our attention to the wet spell duration distribution. Replacing $\sigma$ by $\sigma_1$, and taking $B(u) = \pi B_1(u) + (1-\pi)B_2(u)$ leads to $$F_X(x)=\pi B_1\left(H^{\kappa}_{\xi} (\frac{x}{\sigma_1})\right) + (1-\pi)B_2\left(H^{\kappa}_{\xi} (\frac{x}{\sigma_1})\right).$$
Then using $\xi=0,\; \kappa=1,\;B_1(u)=u, \;B_2(u)=1-\left( 1-u\right)^{\frac{\sigma_1}{\sigma_2}}$  gives:
$$F_X(x)=\pi \left( 1 - \exp(-\frac{x}{\sigma_1}) \right) + (1-\pi) \left( 1 - \exp(-\frac{x}{\sigma_2}) \right).$$ Use $\tau^{(1)} \coloneqq \lceil X\rceil$. For $d\in\{1,2,\dots\}$,
\[
\mathbb{P}\!\left(\lceil X\rceil=d\right)
=\mathbb{P}(d-1<X\le d)
=F_X(d)-F_X(d-1).
\]
we get
\begin{align*}
\mathbb{P}\!\left(\tau^{(1)}=d\right)
&=\pi\left(e^{-(d-1)/\sigma_1}-e^{-d/\sigma_1}\right)
+(1-\pi)\left(e^{-(d-1)/\sigma_2}-e^{-d/\sigma_2}\right)\\
&=\pi\left(1-e^{-1/\sigma_1}\right)e^{-(d-1)/\sigma_1}
+(1-\pi)\left(1-e^{-1/\sigma_2}\right)e^{-(d-1)/\sigma_2}.
\end{align*}
Equivalently, setting $p_i=1-e^{-1/\sigma_i}$ for $i\in\{1,2\}$,
\[
\mathbb{P}\!\left(\tau^{(1)}=d\right)
=\pi\,p_1(1-p_1)^{d-1}+(1-\pi)\,p_2(1-p_2)^{d-1},
\]
which is exactly the expression specified in \eqref{eq:mixgeom_pmf_wet}.

\subsection{Simulation-based Q-Q plots with parametric bootstrap envelopes}
\label{subsec:simulation_based_qqplots_explanation}

Let $(\tau_i^{(r)})_{1\le i\le n}$ denote the recorded spell durations in state $r\in\{0,1\}$, and let $F_{\hat\theta}^{(r)}$ be the fitted parametric distribution for that state. To assess goodness-of-fit, we compare the recorded order statistics with the distribution of order statistics expected under the fitted model. Write $\tau_{(1)}^{(r)}\le \cdots \le \tau_{(n)}^{(r)}$ for the sorted recorded sample. For each bootstrap replicate $b=1,\dots,B$, we generate an i.i.d.\ sample of size $n$ from the fitted model, with parameters held fixed at $\hat\theta$,
$$(\tilde\tau_k^{*(b),(r)})_{1\le k\le n}\sim F_{\hat\theta}^{(r)},$$
and sort it as
$\tilde\tau_{(1)}^{*(b),(r)}\le \cdots \le \tilde\tau_{(n)}^{*(b),(r)}.$ For each rank $k=1,\dots,n$, we then compute the empirical $\alpha/2$ and $1-\alpha/2$ quantiles across bootstrap replicates,
\[
\ell_k^{(r)}=
Q_{\alpha/2}\!\Bigl((\tilde\tau_{(k)}^{*(b),(r)})_{1\leq b \leq B}\Bigr),
\qquad
u_k^{(r)}=
Q_{1-\alpha/2}\!\Bigl((\tilde\tau_{(k)}^{*(b),(r)})_{1\leq b \leq B}\Bigr),
\]
where $Q_p(\cdot)$ denotes the empirical $p$-quantile of a finite sample. Then, $[\ell_k^{(r)},u_k^{(r)}]$ is a pointwise $(1-\alpha)$ bootstrap
envelope for the $k$th order statistic under the fitted model.

In the Q-Q display, the points $(\tau_{(k)}^{(r)},\{\tilde\tau_{(k)}^{*(B),(r)}\})_{1\leq k \leq n}$ are plotted along with the bootstrap envelope.  The area of the points for each pair is proportional to the square root of the number of identical pairs.

\subsection{Goodness-of-fit test on simulated dataset}\label{subsec:gof_simulated_data}

To assess the finite-sample behaviour of the goodness-of-fit statistic in proposition~\ref{prop:gof_q}, we performed a simulation study. Spell durations were simulated from the fitted eGPD distribution, using the estimated
parameters $(\hat f_1,\hat \xi,\hat \sigma,\hat \kappa)$ on the ECAD dataset (Section~\ref{subsec:param_estimation_ecad}). For each simulated sample, we computed the empirical exit probabilities
$\widehat{\mathbf q}^{(r)}$ and evaluated the statistic
\[
Q_{N_n,\; \hat \theta} = N_n\,\mathbf{\Delta}_{\hat \theta}^{\mathsf T}
\left(\mathbf{T}_{\hat \theta}\Sigma_{\hat \theta}
\mathbf{T}_{\hat \theta}^{\mathsf T}\right)^{-1}\mathbf{\Delta}_{\hat \theta},
\]
where $\mathbf{\Delta}$ is the difference between empirical and true probabilities. In the experiment we set $d_{\max}=10$, so that the asymptotic distribution under $H_0^{(r)}$ is $\chi^2_{9}$. Fig.~\ref{fig:gof_simulation} displays (top panel) the histogram of the simulated values of the statistic $Q_n$ (grey bars) together with the theoretical $\chi^2_{9}$ density (black dashed curve), and (bottom panel) the histogram of the corresponding $p$-values together with the density of the standard Uniform distribution (black dashed line). Under the null hypothesis, the $p$-values should be approximately uniformly distributed. The simulation confirms that both distributions are well reproduced, supporting the validity of the test in finite samples.

\begin{figure}[!htbp]
\centering
\includegraphics[width= 0.75\linewidth]{ 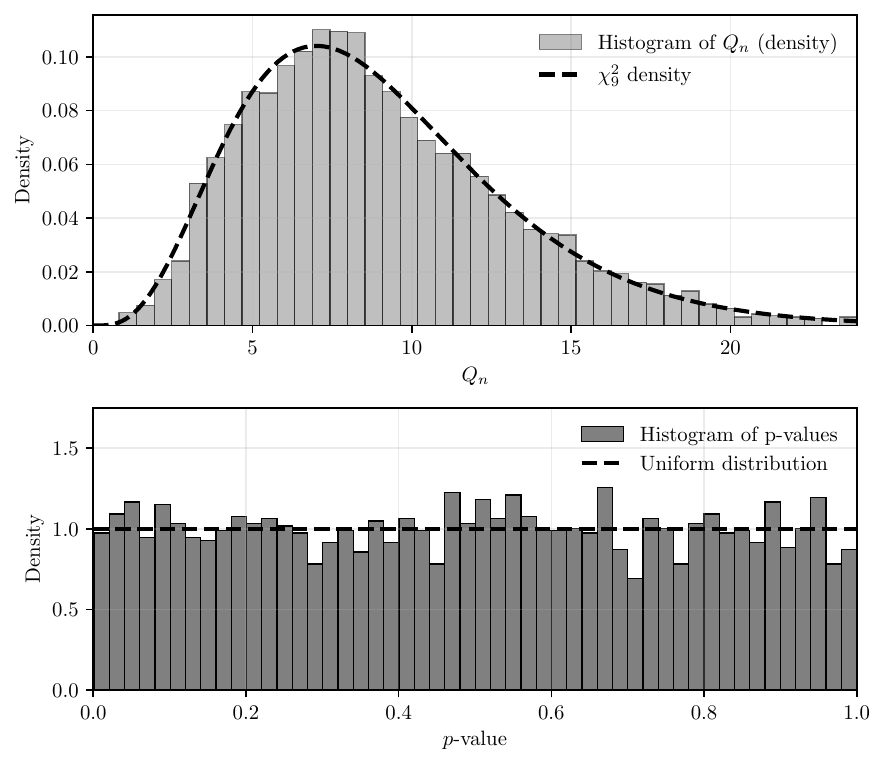}
\caption{Simulation check of the chi-squared goodness-of-fit test described in proposition~\ref{prop:gof_q} on data simulated with a BMCD with hdeGPD distributed spells. Top: histogram (density) of the simulated test statistic $Q_n$ compared with the theoretical $\chi^2_{9}$ density. Bottom: histogram of the associated $p$-values with the Uniform$(0,1)$ density shown as a dashed line.}
\label{fig:gof_simulation}
\end{figure}

\subsection{Mean residual duration estimation}\label{subsec:expected_value_long_dry_spells}

The mean residual duration after \(d\) dry days is defined in equation \eqref{eq:mean_residual_dry_spell_duration}. Since \(\tau^{(0)}\) is a positive discrete random variable,
\begin{equation}\label{eq:expected_val_residual_over_tau_0}
\begin{aligned}
\mathbb{E}\bigl[\tau^{(0)}-d \mid \tau^{(0)}>d\bigr]
&=
\frac{\mathbb{E}\big[(\tau^{(0)}-d)^+\big]}{\overline F_{\tau^{(0)}}(d)}\\
&=
\frac{
\sum_{k=1}^{\infty}\overline F_{\tau^{(0)}}(d+k-1)
}{
\overline F_{\tau^{(0)}}(d)
}\\
&=
\frac{
\mathbb{E}[\tau^{(0)}]-\sum_{k=1}^{d}\overline F_{\tau^{(0)}}(k-1)
}{
\overline F_{\tau^{(0)}}(d)
}.
\end{aligned}
\end{equation}

First let us have an estimation of this quantity for the model specification chosen in this article in Section \ref{subsec:spell_duration_distribution_specification}. Let $X$ be a type-1 eGPD random variable. We have, by construction of $\tau^{(0)}$ probability mass function \eqref{eq:tau0_pmf}:
\begin{equation}\label{eq:link_tau_egpd}
\mathbb E[\tau^{(0)}] = 1+(1-f_1)\,\mathbb E[\lceil X\rceil].
\end{equation}
Since $X$ is nonnegative and continuous, for any integer $u\ge 0$,
\[
\mathbb E[\lceil X\rceil]=\sum_{m=0}^\infty \overline F_X(m)=
\sum_{m=0}^{u-1}\overline F_X(m)+\sum_{m=u}^\infty \overline F_X(m).
\]
For $\xi < 0$, we can have an exact computation because the sum has a finite number of terms, as $\overline F_X (m) =0,$ for $m \ge u_{\lim} \coloneqq -\frac{\sigma}{\xi}$. For $\xi > 0$ we compute bounds. As $\overline F_X$ is nonincreasing, we have
\[
\int_u^\infty \overline F_X(x)\,dx
\le
\sum_{m=u}^\infty \overline F_X(m)
\le
\overline F_X(u)+\int_u^\infty \overline F_X(x)\,dx.
\]
Therefore,
\[
L_u
\le
\mathbb E[\lceil X\rceil]
\le
L_u+\overline F_X(u),
\]
where
\[
L_u:=\sum_{m=0}^{u-1}\overline F_X(m)+\int_u^\infty \overline F_X(x)\,dx.
\]
Combining this with \eqref{eq:expected_val_residual_over_tau_0} and \eqref{eq:link_tau_egpd}, we get
\begin{equation}\label{eq:bounds_value_long_dry_spells_refined}
\begin{gathered}
\frac{
1+(1-f_1)L_u
-\sum_{k=1}^{d}\overline F_{\tau^{(0)}}(k-1)
}{
\overline F_{\tau^{(0)}}(d)
}
\\
\le
\mathbb{E}\bigl[\tau^{(0)}-d \mid \tau^{(0)}>d\bigr]
\\
\le
\frac{
1+(1-f_1)(L_u+\overline F_X(u))
-\sum_{k=1}^{d}\overline F_{\tau^{(0)}}(k-1)
}{
\overline F_{\tau^{(0)}}(d)
}.
\end{gathered}
\end{equation}
As $u \to +\infty$ the width between those bounds converges to $0$ so we can get an approximation of arbitrary precision. These bounds are finite whenever $\xi<1$, which is exactly the condition ensuring that the tail integral in \eqref{eq:int_survival_egpd1_from_p} is finite.

For the purpose of comparison, let us consider a two-state first-order Markov chain
baseline, so that the durations of dry spell are geometrically distributed. Then, using \eqref{eq:expected_val_residual_over_tau_0}, we obtain
\begin{equation}\label{eq:mean_excess_duration_dry_spells_markov_case}
\mathbb{E}\bigl[\tau^{(0)}-d \mid \tau^{(0)}>d\bigr]
=
\frac{
\frac{1}{p_{\mathrm{geom,dry}}}-\frac{1-(1-p_{\mathrm{geom,dry}})^d}{p_{\mathrm{geom,dry}}}
}{
(1-p_{\mathrm{geom,dry}})^d
}
=
\frac{1}{p_{\mathrm{geom,dry}}}.
\end{equation}
The final equality follows directly from the memoryless property of the geometric distribution.

\subsection{Proportion of time in severe dry spell}\label{subsec:approx_proportion_time_long_dry_spell}

The bounds \eqref{eq:bounds_value_long_dry_spells_refined} of Section~\ref{subsec:expected_value_long_dry_spells} are very close to the bounds for the proportion of time in long dry spell described in Example~\ref{ex:long_dry_spell}. One can thus approximate the latter using the bounds:
\begin{equation}\label{eq:final_bounds_proportion_long_dry_spells_refined}
\begin{gathered}
\frac{
1+(1-f_1)L_u-\sum_{k=1}^{d}\overline F_{\tau^{(0)}}(k-1)
}{
1+(1-f_1)(L_u+\overline F_X(u))+\mathbb E[\tau^{(1)}]
}
\\[2mm]
\le
\frac{\mathbb{E}\big[(\tau^{(0)}-d)^+\big]}{\mathbb E[\tau]}
\\[2mm]
\le
\frac{
1+(1-f_1)(L_u+\overline F_X(u))-\sum_{k=1}^{d}\overline F_{\tau^{(0)}}(k-1)
}{
1+(1-f_1)L_u+\mathbb E[\tau^{(1)}]
},
\end{gathered}
\end{equation}
Using all the notations as introduced in Section~\ref{subsec:expected_value_long_dry_spells}, and in the case of the model specification of Section~\ref{subsec:spell_duration_distribution_specification}. As $u \to +\infty$ the width between those bounds converges to $0$ so we can get an approximation of arbitrary precision.

\subsection{Truncated type-1 eGPD expectation}\label{subsec:type1egpd}
Let us start by expressing the first order moment of a continuous type-1 eGPD distributed random variable. Let $X$ be a random variable having an eGPD with parameters $(\kappa,\sigma,\xi)$. When $\xi \neq 0$ and $\xi <1$, the order-1 moment is given by
\hide{\[
\mathbb{E}[X]=
\begin{cases}
\dfrac{\sigma}{\xi}\left(\kappa\,\mathrm{b}(\kappa,1-\xi)-1\right),
& \xi<1,\ \xi\neq 0,\\[8pt]
\sigma\bigl(\psi(\kappa+1)+\gamma\bigr),
& \xi=0,
\end{cases}
\]}

\begin{equation}\label{eq:expected_value_egpd1}
    \mathbb{E}[X]=\dfrac{\sigma}{\xi}\left(\kappa\,\mathrm{b}(\kappa,1-\xi)-1\right)
\end{equation}
where
\[
\mathrm{b}(b_1,b_2)=\int_0^1 t^{b_1-1}(1-t)^{b_2-1}\,dt
\hide{=\frac{\Gamma(b_1)\Gamma(b_2)}{\Gamma(b_1+b_2)}}
\]
denotes the beta function\hide{, $\psi$ is the digamma function, and $\gamma$ is the Euler constant}. This first moment exists if and only if $\xi<1.$ For readability, calculation details are at the end of this Section.

Below are the calculation details for equation:~\eqref{eq:expected_value_egpd1}. Define just for this paragraph a random variable $U\sim \mathrm{Unif}(0,1)$. One may write $X=\frac{\sigma}{\xi}\left[(1-U^{1/\kappa})^{-\xi}-1\right].$
Taking expectations yields
\[
\mathbb{E}[X]
=
\frac{\sigma}{\xi}
\left(
\mathbb{E}\big[(1-U^{1/\kappa})^{-\xi}\big]-1
\right).
\]
Now set $W=U^{1/\kappa}.$ Then $W$ has density $f_W(w)=\kappa w^{\kappa-1},\qquad 0<w<1.$
Hence
\[
\mathbb{E}\big[(1-W)^{-\xi}\big]
=
\kappa\int_0^1 w^{\kappa-1}(1-w)^{-\xi}\,dw.
\]
Recognizing the beta integral $ \mathrm{b}(\kappa,1-\xi) \coloneqq \int_0^1 w^{\kappa-1}(1-w)^{-\xi}\,dw$,
we obtain \eqref{eq:expected_value_egpd1}.
\hide{Using $\mathrm{b}(a,b)=\frac{\Gamma(a)\Gamma(b)}{\Gamma(a+b)},$this may also be written as:
\[
\mathbb{E}[X]
=
\frac{\sigma}{\xi}\left(
\frac{\Gamma(\kappa+1)\Gamma(1-\xi)}{\Gamma(\kappa+1-\xi)}-1
\right).
\]}
For this quantity to be finite, one needs $\xi<1$.
\hide{
In the case $\xi=0$, the distribution reduces to
\[
F_{\kappa,\sigma,0}(x)=\bigl(1-e^{-x/\sigma}\bigr)^\kappa,
\qquad x\ge 0,
\]
and the first moment is obtained by taking the limit as $\xi\to 0$, yielding
\[
\mathbb{E}[X]
=
\sigma\bigl(\psi(\kappa+1)+\gamma\bigr).
\]}

Now let us consider the slightly different quantity

\[
\int_u^\infty \overline F(x)\,dx=\mathbb E[(X-u)_+].
\]
Using the same representation $X=\frac{\sigma}{\xi}\left((1-U^{1/\kappa})^{-\xi}-1\right),
$ with $ U\sim\mathrm{Unif}(0,1)$, and we still use $W=U^{1/\kappa}$, so that $f_W(w)=\kappa w^{\kappa-1}$ on $(0,1)$, one gets
\[
\int_u^\infty \overline F(x)\,dx
=
\mathbb E\!\left[\left(\frac{\sigma}{\xi}\big((1-W)^{-\xi}-1\big)-u\right)_+\right].
\]
We have $X>u$ is equivalent to $W>a_u$, where $a_u:=1-\left(1+\frac{\xi u}{\sigma}\right)^{-1/\xi}$.
Hence
\[
\int_u^\infty \overline F(x)\,dx
=
\int_{a_u}^{1}
\left[
\frac{\sigma}{\xi}\big((1-w)^{-\xi}-1\big)-u
\right]\kappa w^{\kappa-1}\,dw.
\]
Splitting the integral gives
\[
\int_u^\infty \overline F(x)\,dx
=
\frac{\sigma\kappa}{\xi}\int_{a_u}^{1} w^{\kappa-1}(1-w)^{-\xi}\,dw
-
\left(\frac{\sigma}{\xi}+u\right)\int_{a_u}^{1}\kappa w^{\kappa-1}\,dw.
\]
Denoting $b_x$ the incomplete beta function, defined as
\[
b_x(b_1,b_2)\coloneqq\int_0^x t^{b_1-1}(1-t)^{b_2-1}\,dt,
\]
one has
\[
\int_{a_u}^{1} w^{\kappa-1}(1-w)^{-\xi}\,dw
=
\mathrm{b}(\kappa,1-\xi)-b_{a_u}(\kappa,1-\xi),
\]
while
\[
\int_{a_u}^{1}\kappa w^{\kappa-1}\,dw
=
1-a_u^\kappa.
\]
Therefore
\begin{equation}\label{eq:int_survival_egpd1_from_p}
\int_u^\infty \overline F(x)\,dx
=
\frac{\sigma\kappa}{\xi}\Bigl(\mathrm{b}(\kappa,1-\xi)-b_{a_u}(\kappa,1-\xi)\Bigr)
-\left(\frac{\sigma}{\xi}+u\right)\bigl(1-a_u^\kappa\bigr)
\end{equation}
with
\[
b_x(b_1,b_2)=\int_0^x t^{b_1-1}(1-t)^{b_2-1}\,dt,
\qquad
a_u=1-\left(1+\frac{\xi u}{\sigma}\right)^{-1/\xi}.
\]
The incomplete beta function is implemented in many packages. Just as earlier, this quantity is finite only when $\xi<1$.

\end{appendices}

\end{document}